\documentclass[10pt,aps,pre,preprint,superscriptaddress]{revtex4-1}

\usepackage{amsmath,amsthm}
\usepackage{amssymb}
\usepackage{amscd}
\usepackage{wasysym}
\usepackage[ansinew]{inputenc}
\usepackage[T1]{fontenc}
\usepackage{ae,aecompl}
\usepackage{hyperref}
\usepackage{enumitem}

    \usepackage[pdftex]{graphicx}

\usepackage{color}
\definecolor{red}{rgb}{1,0,0}
\definecolor{blue}{rgb}{0,0,1}
\definecolor{green}{rgb}{0,0.5,0}
\definecolor{magenta}{rgb}{1,0,1}

\newcommand{\ceil}[1]{\left\lceil#1\right\rceil}

\newtheorem{thm}{Theorem}
\newtheorem{defn}[thm]{Definition}
\newtheorem{remark}[thm]{Remark}

\newcommand{\be}{\begin{equation}}
\newcommand{\ee}{\end{equation}}
\newcommand{\bea}{\begin{eqnarray}}
\newcommand{\eea}{\end{eqnarray}}
\newcommand{\nn}{\nonumber}

\newcommand{\N}{{\cal N}}

\renewcommand{\vec}[1]{\boldsymbol{#1}}
\newcommand{\R}{\mathcal{R}}
\newcommand{\F}{\mathcal{F}}

\newcommand{\floor}[1]{\left\lfloor #1 \right \rfloor}

\newcommand{\cfdom}{\mathcal{D}}
\newcommand{\cfdomlim}{\tilde{\mathcal{D}}}
\newcommand{\DualGraph}[1]{G^{\rm dual}(#1)}
\newcommand{\LoopyLapl}{\tilde{L}^{\rm loopy}}
\newcommand{\MetaGraph}[1]{\tilde{#1}}
\newtheorem{lemma}{Lemma}

\newtheorem{corr}{Corollary}

\begin{document}
\bibliographystyle{apsrev}
\title{Cycle flows and multistabilty in oscillatory networks: an overview}

\author{Debsankha Manik}
\affiliation{Network Dynamics, Max Planck Institute for Dynamics and Self-Organization (MPIDS), 37077 G\"ottingen, Germany}

\author{Marc Timme}
\affiliation{Network Dynamics, Max Planck Institute for Dynamics and Self-Organization (MPIDS), 37077 G\"ottingen, Germany}
\affiliation{Faculty of Physics, Georg August University G\"ottingen, 37073 G\"ottingen,
  Germany}
\affiliation{Technical University Darmstadt, 64289 Darmstadt, Germany}

  \author{Dirk Witthaut}
\affiliation{Forschungszentrum J\"ulich, Institute for Energy and Climate Research -
	Systems Analysis and Technology Evaluation (IEK-STE),  52428 J\"ulich, Germany}
\affiliation{Institute for Theoretical Physics, University of Cologne, 
		50937 K\"oln, Germany}

\date{\today }

\begin{abstract}
The functions of many networked systems in physics, biology or engineering rely
on a coordinated or synchronized dynamics of its constituents. In power grids for
example, all generators must synchronize and run at the \emph{same frequency} and their phases need to
appoximately lock to guarantee a steady power flow. Here, we analyze the
existence and multitude of such phase-locked states. Focusing on edge and cycle
flows instead of the nodal phases we derive rigorous results on the existence
and number of such states. Generally, multiple phase-locked states coexist
in networks with strong edges, long elementary cycles and a homogeneous distribution
of natural frequencies or power injections, respectively. We offer an algorithm to
systematically compute multiple phase-locked states and
demonstrate some surprising dynamical consequences of multistability.

\end{abstract}

\maketitle


\section{From Kuramoto oscillators to power grids}

Coupled oscillator models are ubiquitous in science and technology,
describing the collective dynamics of various systems in
micro to  macro scale. Research on coupled
oscillators dates back to Christian Huygens, who noticed that
two clocks synchronize when they are coupled \cite{Huyg93}.
One of the most important mathematical models was introduced
by Kuramoto \cite{Kura75,Kura84} and successfully applied
to describe the collective dynamics 
of coupled Josephson junctions \cite{Wies96},
neuronal networks \cite{Vare01},
chemical oscillators \cite{Kiss02},
and a variety of other synchronization phenomena
\cite{Stro00,Aceb05,Aren08,dirk2014kuramoto}.

That model  \cite{Kura84} describes the dynamics of $N$ coupled limit cycle 
oscillators. The equations of motions for the 
phases $\theta_j, \; j \in \left\{1,\ldots,N\right\}$ are given by
\be
   \frac{d}{dt} \theta_j = \omega_j  + \sum_{\ell = 1}^N
      K _{j,\ell} \sin(\theta_\ell - \theta_j).
   \label{eqn:kuramoto}
\ee 
The coupling matrix is assumed to be symmetric, $K _{j,\ell} = K _{\ell,j}$ and the
$\omega_j$ are the natural frequencies of the oscillators. Throughout this 
article we consider systems where all
$K _{j,\ell} \geq 0$.

A similar model of second-order oscillators describes the collective phenomena of 
animal flocks \cite{Erme91,Ha10} or human crowds \cite{Stro05} as well as the 
coarse-scale dynamics of power grids \cite{Berg81,Fila08,12powergrid,Mott13,Dorf13,13powerlong,schafer2015decentral}.
For power grids, for instance, the units $j$ describe synchronous machines, generators or motors, 
whose state is completely described by their phase $\theta_j$ and the phase 
velocity $\dot \theta_j$ relative to the reference frequency of the grid, 
typically  rotating at 50 Hz or 60 Hz.  
The acceleration (deceleration) of the machines is proportional to the sum of 
the mechanical 
power $P_j$ generated (consumed) by the machine including damping and
the electric power exchanged with the grid. The detailed equations of motion 
are given by
\be
   M_j \frac{d^2}{dt^2} \theta_j + D_j \frac{d}{dt} \theta_j 
          = P_j   + \sum_{\ell = 1} K _{j,\ell} \sin(\theta_\ell - \theta_j),
  \label{eqn:power}
\ee
where $M_j$ is an inertia term and $D_j$ a damping constant.
The coupling constants $K_{j,\ell} = U^2 B_{j,\ell}$ are determined
by the voltage $U$ of the grid, which is assumed to be constant, and the admittance $B_{j,\ell}$
of the electrical transmission line joining node $j$ and node $\ell$.  
The flow of electric real power from node $\ell$ to node $j$ is
\be
   \label{eq:def-flow} 
    F_{j,\ell}= K _{j,\ell} \sin(\theta_\ell - \theta_j)  =  K _{j,\ell}  S_{j,\ell}.
\ee
It is useful to describe the interaction topology of the system as a weighted graph $G(V,E)$, whose vertex set $V$ 
is identical with the set of oscillators, and edge set $E$ is given by the set of all inter-oscillator 
coupling pairs, i.e., all pairs with $K _{\ell,j} > 0$. We use the term network \cite{Newm10} (rather than 
the term graph)  for the entire system with given natural frequencies $\omega_j$ or the powers $P_j$.

Here we distinguish two types of synchronization in oscillator networks.
Traditionally, the emergence of partial synchrony has received the most
interest in the physics community \cite{Kura75,Kura84,Stro00,Aceb05}.
In his seminal work, Kuramoto investigated a set of oscillators with global coupling,
$K_{j,\ell} = K/N$ and natural frequencies drawn at random from a unimodal symmetric 
distribution $g(\omega)$. If the coupling constant $K$ exceeds a critical value $K_c$, 
a fraction of the oscillators start to synchronize in the sense that they rotate at the 
same angular velocity, although their natural frequencies differ. 
In this state, called ``frequency locking'', the phases of parts of the oscillators are ordered, but they are not
strictly phase-locked, such that the phase difference of two synchronized
oscillators $(\theta_j - \theta_\ell)$ is generally small, but not constant.

In this article, we analyze the properties of \emph{globally phase-locked states},
where all oscillators synchronize and the phase differences $(\theta_j - \theta_\ell)$ 
are constant for all pairs $(j,\ell)$.  These states are especially important for
power grids, as they describe the regular synchronous operation of the grid
\cite{Berg81,Fila08,12powergrid,Mott13,Dorf13}.
If this state is lost due to local outages or accidents, the grid will fragment into
asynchronous islands which can no longer exchange electric energy \cite{witthaut2016critical}. 
For instance, the European power grid fragmented into three asynchronous 
areas on November 4th 2006 after the shutdown of one transmission line
in Northern Germany. As a result, south-western Europe suffered an under-supply 
on the order of 10 GW and approximately 10 million households were disconnected
\cite{UCTE07}.

Without loss of generality, we take $\sum_j \omega_j = 0$ or 
$\sum_j P_j = 0$ respectively,  by invoking a transformation
to a co-rotating frame of reference.  The globally phase-locked states are then 
the \emph{fixed points} of the system. Both for the Kuramoto model and 
the power grid model, these states are given by the solutions of the algebraic equations
\be
       P_j  + \sum_{\ell = 1}^N
      K _{j,\ell} \sin(\theta_\ell - \theta_j) =0 \qquad
     \text{for all } j \in \left\{1,\ldots, N\right\},
    \label{eqn:def-steady-state}
\ee
replacing $P_j$ by $\omega_j$ for the Kuramoto model.
In the following we analyze the influence of the network
topology given by the coupling matrix $K_{j,\ell}$ on the
existence of a fixed point. All results below hold
for both models, nevertheless our intuition heavily relies
on the interpretation of $F_{j,\ell}= K _{j,\ell} \sin(\theta_\ell - \theta_j)$
as a flow which is inspired from the power grid model. 
The results generalize  to  arbitrary coupling
functions $f$ instead of the sine (see, e.g., \cite{Wata93,Bick11}). 
In the following we mostly restrict ourselves to the common 
sine coupling for the sake of clarity.

\section{The nature and bifurcations of fixed points}
\label{sec:bifurcation}

Both the Kuramoto system and the oscillator model of power grids share the same 
set of fixed points \eqref{eqn:def-steady-state}. It has 
been shown that the similarity between these two systems runs deeper, namely, 
the linear 
stability properties of those fixed points are identical \cite{epjst14}.
In this section we briefly review some basic results on the stability
of the fixed points.

We analyze the dynamical stability of a certain fixed point 
$\vec{\theta^*}=(\theta_1^*,\ldots,\theta_N^*)$ 
by defining the potential function
\begin{align}
\label{eq-potential}
V(\theta_1, \theta_2,\ldots, \theta_N)=- \sum_j P_j \theta_j - \frac{1}{2} \sum_{i,j} K_{ij} \cos(\theta_i - \theta_j).
\end{align}
The fixed points correspond to the local extrema of this potential, where 
$\frac{\partial V}{\partial \theta_j} = 0 \text{ for all } j$.   A fixed
point $\vec{\theta}^*$ is asymptotically stable if the Hesse matrix $M$ of the potential 
function 
\begin{align}
\label{eq-M}
   M(\vec{\theta^*})&=\begin{pmatrix}
      \sum_\ell K_{1,\ell}^{\rm red} & 
              - K_{1,2}^{\rm red} & \cdots \\
     - K_{2,1}^{\rm red} & 
               \sum_l K_{2, \ell}^{\rm red} & \cdots \\
       \vdots & \vdots  & \ddots 
  \end{pmatrix}
\end{align}
with the \emph{residual capacity}
\be
   \label{eq:kred}
   K_{j,\ell}^{\rm red} = K_{j,\ell} \cos{(\theta_j^*-\theta_\ell^*)} 
\ee
has positive eigenvalues only. It is worth noting that $M$ has one eigenvector 
$\vec{v}_1=(1,1,\cdots,1)$ with eigenvalue $\mu_1 = 0$, because any fixed point
$\vec{\theta^*}$ is arbitrary up to an 
additive constant $c$. As such a global phase shift does not affect the locking of 
the phases we can discard it in the following and concentrate on the stability 
transversely to the solution space 
$\{ \vec{\theta^*} + c (1,1,\cdots,1) | c \in \mathbb{R} \}$.

\begin{lemma}
\label{thm:stability-M}
Let the eigenvalues of $M$ be ordered such that $\mu_1=0$ and 
$\mu_2\leq\cdots\leq \mu_N$.  
If for a given network topology and a given fixed point, 
\be
    \label{def:eigval_mu}
    \mu_k > 0, \quad \text{ for all } k\in\left\{2,3,\cdots,N\right \}, 
\ee
then this fixed point is transversally \emph{asymptotically stable} for both 
Kuramoto system and the power grid model system. 
If one of the $\mu_k < 0$, then the dynamical system 
is linearly unstable.
\end{lemma}
Using some results from bifurcation theory, it has been shown in
\cite{epjst14} that a stable fixed point can only be lost by an inverse 
saddle-node bifurcation when one of the eigenvalues becomes zero,
$\mu_2 = 0$. At this point linear stability analysis is not sufficient to
predict stability of the fixed point but it is expected that the fixed point 
is unstable \cite{khazin1991stability}.

More insights can be gained about the loss of a fixed point when 
the phase differences across all edges in the network are sufficiently 
small:
\begin{corr}
\label{corr:stab-phasediff}
Consider a simply connected network. A fixed point $\vec{\theta^*}$ is transversally asymptotically stable if 
\begin{align}
  \label{def:normalop}
  \cos{\left(\theta_i^*-\theta_j^*\right) > 0}
\end{align}
holds for all edges $(i,j)$ in the network. Then the network is said to be in ``normal operation''.
\end{corr}

\begin{proof}
To this end, we first define a meta graph as follows.  
\begin{defn}[Meta graph]
    \label{def:mata-graph}
    Given a graph $G(V, E)$, and a set of flows $F_{uv}$ across each edge $e(u,v)$,
    its meta graph $\MetaGraph{G}$ is an undirected graph with vertex set $V$ 
    and edge set $E'$ defined as follows.  For all edges $e(u,v)\in E$, with 
    weight $K_{uv}, \exists$ an edge $e(u,v)\in E'$ with weight 
    $K_{uv}^{\rm red} = \sqrt{K_{uv}^2-F_{uv}^2}$, as per \eqref{eq:kred}.
\end{defn}
Then the matrix $M$ as defined in \eqref{eq-M} is seen to be the Laplacian matrix 
of the meta graph $\tilde{G}$. The eigenvalues of a Laplacian of a connected 
undirected graph
with positive edge weights are always non-negative \cite{Newm10} such that
we obtain the result.
\end{proof}

During normal operation an eigenvalue of $M$ can become $0$ only when $\tilde{G}$ 
disconnects into two (or more) components. Such a split-up happens only when
$K_{j,\ell}^{\rm red} = 0$ for all the transmission lines connecting two 
certain parts (denoted by $G_1, G_2$) of the network, meaning that these lines
are completely saturated
\be
  \sin{\left(\theta_j^*-\theta_\ell^*\right)} = \pm 1  \quad \Rightarrow \quad |F_{j,\ell}| = K_{j,\ell}
    \qquad \text{ for all }  (j,\ell)\in E, j\in G_1,\ell \in G_2. 
\ee
Another scenario for the loss of stability is that one or more transmission lines
leave normal operation. Then the edge weights become effectively negative, such
that a simple graph-theoretic interpretation of the bifurcation is no longer
possible \cite{epjst14}.

\section{Cycle flows and Geometric frustration}
\label{sec:frustration}

\subsection{Flow conservation and the dynamics condition}

It is instructive to divide the defining equation 
(\ref{eqn:def-steady-state}) of a 
fixed point into two parts.
First, every fixed point has to satisfy a dynamic condition 
which is nothing but the conservation of the flow at every node
of the network
\begin{subequations}
\label{eqn:dc1}
\begin{align}
   & P_j + \sum_{\ell=1}^N K _{j,\ell} S_{j,\ell} = 0 \qquad 
    \text{ for all } j \in \{1,\ldots,N \}\\
  &  |S_{j,\ell}|   \le 1 \quad \qquad \qquad \qquad  
              \text{ for all edges }  (j,\ell).
\end{align} 
\end{subequations}
Here, $\sum_\ell K_{j,\ell}S_{j,\ell}$ is the sum of all 
flows from the neighboring nodes to the node $j$, while 
$P_j$ is a source or sink term, respectively.
The second part of this condition reflects the fact that the
transmission capacity of each link is bounded, such that the
magnitude of the flow $|F_{j,\ell}|$ cannot exceed the capacity
$K_{j,\ell}$.  
The dynamic condition (\ref{eqn:dc1}) holds for
all flow networks including also DC networks (i.e. Kirchhoff's
rules) and biological network models \cite{Kati10, ronellenfitsch2015dual}. 

To obtain a better understanding of the possible solutions,
we slightly rephrase the dynamic condition  (\ref{eqn:dc1}).
In particular, we label all the $L$ edges in the network with
$e \in \left\{ 1,\ldots,L\right\}$.
As the flows are directed, we have to keep track of the ordering
of the vertices connected by the edge $e$. That is, each $e$ 
corresponds to a directed link $(j,\ell)$ in the following. The ordering
is arbitrary but must be kept fixed.
Then we write $S_e = S_{j,\ell}$ and $F_e = F_{j,\ell}$
for the flow over a link $e  \,  \widehat= \, (j,\ell)$. 
Furthermore, we define the unweighted edge incidence matrix
$I \in \mathbb{R}^{N\times L}$ \cite{Newm10} via
\be
   I_{j,e} = \left\{
   \begin{array}{r l }
      +1 & \; \mbox{if node $j$ is the head of edge $e  \,  \widehat= \, (j,\ell)$},  \\
      -1 & \; \mbox{if node $j$ is the tail of edge $e  \,  \widehat= \, (j,\ell)$},  \\
      0     & \; \mbox{otherwise}.
  \end{array} \right.
\ee
and the weighted edge incidence matrix $\widetilde K \in \mathbb{R}^{N\times L}$ 
with the components $\widetilde K_{je} = K_e I_{je}$.

The dynamic condition (\ref{eqn:dc1}) then reads
\begin{subequations}
\label{eqn:dc3}
\begin{align}
    P_j + \sum_{e=1}^L  I _{j,e} F_e  = 0&  
    \qquad\text{ for all } j=1,\ldots,N 
     \label{eqn:dc3a}  \\
     |F_{e}|    \le K_e & 
     \qquad\text{ for all } e = 1,\ldots,L.
  \label{eqn:dc3b}
\end{align} 
\end{subequations}
in terms of the flows or 
\begin{subequations}
\label{eqn:dc2}
\begin{align}
   & P_j + \sum_{e=1}^L \widetilde K _{j,e} S_e = 0 \qquad 
              \text{ for all } j=1,\ldots,N 
     \label{eqn:dc2a}  \\
  &  |S_{e}|   \le 1 \quad \qquad \qquad \qquad  
              \text{ for all } e = 1,\ldots,L.
  \label{eqn:dc2b}
\end{align} 
\end{subequations}
in terms of the sine factors.
Here, $\vec F = (F_1,\ldots,F_L)^T$ and $\vec S = (S_1,\ldots,S_L)^T$ are vectors 
in $\mathbb{R}^L$.
The matrix $\widetilde K$ has $N$ rows, but its rank is only $(N-1)$.
This is due to the fact that the sum of all rows is zero as 
$\sum  \nolimits_j \widetilde K_{j,e} = 0$, since each edge 
has exactly one head and one tail. Hence, the solutions 
of the linear set of equations (\ref{eqn:dc2a}) span an affine 
subspace of $\mathbb{R}^L$ whose dimension  is $(L-N+1)$.  This statement will 
later be rigorously proved in Lemma \ref{lem:cf-fp-corr}.
In many important applications $L$ is much larger than the number of 
nodes $N$, such that we have a high dimentional submanifold $\mathbb{B}$ of $\mathbb{R}^L$
with every $\vec{S}\in \mathbb{B}$ a solution of \eqref{eqn:dc2}, and hence, a fixed 
point of \eqref{eqn:kuramoto} and \eqref{eqn:power}. 
However, the set of solutions of the dynamical equations
can also be empty if the capacities $K_{j,\ell}$ are too small.
In fact, the condition (\ref{eqn:dc2b}) defines a bounded convex 
polytope in $\mathbb{R}^L$. The solution of the full dynamical
conditions (\ref{eqn:dc2}) are given by the intersection of this
polytope and the $(L-N+1)$ dimensional affine subspace.

We can further characterize the solution of the dynamic conditions, by 
establishing that the homogeneous solutions of the system (\ref{eqn:dc2a}) are 
just the \emph{cycle flows} which do not affect flow conservation. As the number
of fundamental cycles in a network is  $(L-N+1)$, the dimension
of the solution space is also given by $(L-N+1)$.  Derivation of these results 
follows.  

\begin{defn}[Simple cycle]
    \label{def:simplecycle}
    Given an undirected  graph $G(V,E)$, a closed path 
    $c = (v_1, v_2, \cdots,v_l,v_1)$ where no vertex apart from $v_1$ occurs 
    twice is called a \emph{simple cycle} \cite[p~21]{Dies10}. 
\end{defn}

\begin{defn}[Cycle basis]
	\label{def:cyclebasis}
Given a connected graph $G(V, E)$ with $L$ edges and $N$ vertices, the set
of all simple cycles $\mathfrak{C}$
forms a vector space over the two element field 
$GF(2) = \left\{0, 1\right\}$, 
with set symmetric difference being the addition operator. 
This vector space has dimension $L-N+1$.  A basis $B_C$ of this vector space is called a \textbf{cycle 
basis} of the graph $G$.  
\end{defn}

\begin{defn}[Signed characteristic vector of a cycle]
    \label{def:signed_char_vec}
    An arbitrary assignment of a direction to each edge of an undirected graph 
    $G$, which results in a directed graph, is called an orientation 
    $G^{\sigma}$ \cite{godsil2013}. Given a graph $G$ with $L$ edges and $N$ vertices, and one 
    such orientation, there exists an 
    injective mapping from the set $\mathfrak{C}$ of all simple cycles of
    $G$ to $\mathbb{R}^{L}$ as follows:
\begin{align*}
    \mathfrak{C} & \to \mathbb{R}^{L}\\
    c & \mapsto \vec{z^c}\\
    z_e^c & = 
    \begin{cases}
        0, & \text{ if } e \text{ is not in } c\\
        1, & \text{ if } e=(v_i,v_{i+1}) \text{ and } u_{i+1} \text{ is the 
    head of } e\\
    -1, & \text{ if } e=(v_i,v_{i+1}) \text{ and } u_{i+1} \text{ is the tail of 
}e.  
    \end{cases}
\end{align*}
$\vec{z^c}$ is called the \emph{signed characteristic vector} of each cycle.   
\end{defn}

Now we show that any fixed point of the system can be uniquely specified by a
\emph{cycle flow} along each cycle belonging to a cycle basis of the 
underlying graph, alongwith an arbitrary solution of \eqref{eqn:dc3}.  

\begin{defn}[Cycle flow]
    Given a simple cycle $c=(v_1, v_2,\cdots,v_l, v_1)$ belonging to an undirected graph $G(V, E)$, a 
    flow $\vec{F}$ is called a \emph{cycle flow} if 
    \be
    F_{j,k}  = 
            \begin{cases}
                c & \text{ if } (j,k) \in\{(v_1,v_2), (v_2,v_3),\cdots,(v_{l-1},v_l),(v_l,v_1)\}\\
                0 & \text{ otherwise},
            \end{cases}
    \ee
    i.e.  it is a contant nonzero flow along the cycle.  
\end{defn}

\begin{lemma}
\label{lem:cf-fp-corr}
Let $\mathbb{S}_G$ be the set of all fixed points of a network $G$ 
satisfying the normal operation criteria \eqref{def:normalop}.
Then there exists a one-to-one function 
$\vec{f_c}:\mathbb{S}_G \mapsto \mathbf{R}^{L-N+1}$ 
that maps each fixed point to a \emph{cycle flow vector}.  
\end{lemma}

\begin{proof}
    Let $\vec{\theta^{(0)}}$ be one (arbitrarily chosen) fixed point. Let $\vec{\theta}$ be 
    another. Then we construct the mapping $f_c$ by proving that the flows for 
    these two fixed point differ only by \emph{cycle flows} along each cycle.   

    Let $\vec{F^{(0)}}=(F^{(0)}_{e_1},F^{(0)}_{e_2},\cdots,F^{(0)}_{e_L})$ and 
    $\vec{F'}=(F'_{e_1},F'_{e_2},\cdots,F'_{e_L})$ be the flows for the fixed 
    points $\vec{\theta^{(0)}}$ and $\vec{\theta}$, respectively. Then 
    \begin{align}
    \label{eq:flowdiff_cycle_zero}
    \vec{F} - \vec{F^{(0)}} &= \sum_{c \in B_C} f_c \vec{z^c},
    \end{align}
    due to the result from graph theory that the flow space of an 
    oriented graph $G^{\sigma}$ is spanned by the signed characteristic 
    vectors (Definition \ref{def:signed_char_vec}) of its cycles \cite[p~311]{godsil2013}. Since by definition the cycles in $B_C$ forms a basis of the cycle space, 
    the coefficients $f_c$ are guaranteed to be unique. This concludes the 
    proof. 
\end{proof}

\subsection{The winding number and the geometric condition}  

In addition to the dynamic condition, there is a geometric condition for the existence of a fixed point:
a fixed point exists if the flows 
$F_{j,\ell} = K_{j,\ell} S_{j,\ell}$
satisfy the dynamic condition (\ref{eqn:dc2}) 
\emph{and if}  
\bea
    \text{ for all } \, \mbox{edges} \, (\ell,j): \quad \,  \exists \, (\theta_1, 
    \ldots, \theta_N) \mbox{ such that}  \;
    S_{j,\ell} = \sin(\theta_\ell - \theta_j).
  \label{eqn:gc1}
\eea

We now rephrase this condition in a more instructive way.
To this end we assume that we have already obtained a solution
of the dynamic condition (\ref{eqn:dc2}). Then we can try to
successively assign a phase $\theta_j$ to every node $j$ in the
network. Starting at a node $j_0$ with an arbitrary phase
$\theta_{j_0}$, we assign the phases of all neighboring nodes
$j_1$ such that $\sin(\theta_{j_1}   -   \theta_{j_0}) = S_{j_0,j_1}$.
We then proceed in this way through the complete network
to assign the phase of an arbitrary node $j_n$,
\be
   \theta_{j_n} = \theta_{j_0} + \sum_{s=0}^{n-1}
      \Delta_{j_s,j_{s+1}},
   \label{eqn:phase-path}
\ee
where $(j_0,j_1,\ldots,j_n)$ is an arbitrary \emph{path} form $j_0$ to $j_n$
and we have used a solution of the equation 
\be
   S_{j,\ell} = \sin(\Delta_{j,\ell})
   \label{eq:D-Delta-0}
\ee
for every edge $(j,\ell)$.

In general, a given node $j_n$ can be reached
from $j_0$ via a multitude of different paths. To define a unique
set of phases that satisfies the geometric condition (\ref{eqn:gc1}),
we must assure that Eq.~(\ref{eqn:phase-path}) yields a unique 
phase regardless of which path is taken from $j_0$ to $j_n$.
This is equivalent to the condition, that the phase 
differences over every \emph{simple cycle} (as defined in Definition 
\ref{def:simplecycle})  in the network 
must add up to an integer multiple of $2\pi$.  
\bea
       \label{eqn:gc2}
       && \sum_{(j,\ell) \in \text{cycle} \, c} \Delta_{j,\ell} = 2 m 
       \pi,\quad \text{ for some } m\in\mathbb{Z},
\eea
where $\Delta_{j,\ell}$ is a solution of equation (\ref{eq:D-Delta-0}).
Furthermore, it is sufficient if \eqref{eqn:gc2} is satisfied by the cycles in 
the cycle basis of the network  defined in 
Definition \ref{def:cyclebasis}: it will then automatically be satisfied for all
simple cycles of the network, since the simple cyles form a vector space.

However, there are two distinct solutions
\begin{subequations}
    \label{eqn:deltaS}
    \begin{align}
        \Delta_{j,\ell}^+ &= \arcsin(S_{j,\ell}) \label{eqn:deltaS1}\\
        \Delta_{j,\ell}^- &= \pi -  \arcsin(S_{j,\ell}) \label{eqn:deltaS2}
    \end{align}
\end{subequations}
of equation (\ref{eq:D-Delta-0}) that satisfy $\Delta_{j,l}^{\pm}\in [-\pi,\pi)$.  
To consider both, we define a partition of the 
edge set 
\begin{align}
\label{eq:edge_partdef}
E & = E_+ \cup E_- \\
E_+ & =\{(j,\ell)\in E |\Delta_{j,\ell} = \Delta_{j,\ell}^+\}\\
E_- & =\{(j,\ell)\in E |\Delta_{j,\ell} = \Delta_{j,\ell}^-\}.  
\end{align}
Alternatively, we can define the two sets in terms of the 
nodal phases as
\begin{align}
   E_+ &= \{ (i,j) \in E | \cos(\theta_i - \theta_j) > 0 \} \\
   E_- &= \{ (i,j) \in E | \cos(\theta_i - \theta_j) \le 0 \}.
\end{align}
We note that a fixed point where the plus sign is realized for all
edges ($E_- = \emptyset$) is guaranteed to be linearly stable according
to corollary \ref{corr:stab-phasediff}. We refer to it as \emph{normal operation}.

To operationalize the geometric condition we now define the winding number 
\eqref{eq:def-winding}, following the notation used by Ochab and Gora \cite{Ocha09}.
\begin{defn}[Winding vector]
\label{def:windnum}
Consider a connected network with flows $\vec F$. For every fundamental cycle $c$, the winding number with respect to a partition $E = E_+ + E_-$ is defined as
\be
   \label{eq:def-winding}
   \varpi_c = \sum_{e \in E} z^c_e \Delta_e (F_e) 
\ee 
with  
\be
\label{eq:delta}
    \Delta_e (F_e)  = \left\{ \begin{array}{l l l}
        \arcsin(F_e/K_e) & \; {\rm for} \;  & e \in E_+ \\
        \pi - \arcsin(F_e/K_e) & \;  & e \in E_- .
      \end{array} \right.  
\ee   
The winding vector is defined as 
\be
\label{def:wvec}
\vec \varpi = (\varpi_1,\ldots,\varpi_{L-N+1})^T.
\ee
\end{defn}

Using the winding number we can reformulate the conditions for the
existence of a fixed point and establish a correspondence between the 
description of a fixed points in terms of nodal phases of edge flows.
\begin{thm}
\label{th:fp-dyn-geo}
Consider a connected network with power injections $\vec P \in \mathbb{R}^{N}$ and coupling matrix $K \in \mathbb{R}^{N\times N}$. Then the following two statements are equivalent:
\begin{enumerate}
\item $\vec \theta^*$ is a fixed point, i.e., a real solution of equation \eqref{eqn:def-steady-state}.
\item $\vec F \in \mathbb{R}^L$ satisfies the dynamic condition (\ref{eqn:dc3})
   and $\vec \varpi \in \mathbb{Z}^{L-N+1}$ for some partition $E = E_+ + E_-$.
\end{enumerate}
\end{thm}
\begin{proof} We prove the theorem in two parts.\\
    \begin{description}
        \item [(1) $\implies$ (2)]  If $\vec{\theta^*}$ is a fixed point, then the flows $\vec F$ satisfying the 
dynamical condition \eqref{eqn:dc3} as given by 
\eqref{eq:def-flow} are
\be
    F_{j,k} = K_{j,k}\sin{(\theta_k - \theta_j)}.
\ee
Let's partition the edge set into $E^+$ and $E^-$ by
\begin{align}
    \label{eq:E_partition}
    e=(j,k) \in 
    \begin{cases}
        E^+ & \text{ if } \cos{(\theta_k - \theta_j)}  > 0\\
        E^- & \text{ if } \cos{(\theta_k - \theta_j)}  \leq 0.
    \end{cases}
\end{align}

We note the identity that
\begin{align}
    \label{eq:arcsin_sin}
    \arcsin{\left(\sin{(x)}\right)} = 
            \begin{cases}
                -x + (2m_x+1)\pi & \text{ if } \cos{(x)} \leq 0\\
                x + 2m_x\pi & \text{ if } \cos{(x)} > 0,\text{ for some } m_x\in \mathbb{Z}.
            \end{cases}
\end{align}
Combining this identity with the definition of $\Delta_e$ in \eqref{eq:delta} 
and our chosen set partition \eqref{eq:E_partition}, results in
\begin{align}
    \label{}
    \text{ for all } (j,k)\in E^+,\quad  \Delta_{j,k} & = 
    \arcsin{\left(F_{jk}/K_{jk}\right)} \nn\\
                                   & = \arcsin{\left(\sin{(\theta_k - 
                                   \theta_j)}\right)} \nn\\
                                   & = 2m_{jk}\pi + (\theta_k - \theta_j) 
                                   \label{eq:Delta_part1} \\
    \text{ for all } (j,k)\in E^-,\quad  \Delta_{j,k} & = \pi - \arcsin{\left(F_{jk}/K_{jk}\right)}\nn\\
                                   & = \pi - \arcsin{\left(\sin{(\theta_k - \theta_j)}\right)}\nn\\
                                   & = -2m_{jk}\pi + (\theta_k - \theta_j)\label{eq:Delta_part2}.
\end{align}
Combining \eqref{eq:Delta_part1} and \eqref{eq:Delta_part2}, we obtain 
$\Delta_{jk} = 2m_{jk}\pi + (\theta_k - \theta_j), m_{jk}\in \mathbb{Z}$ 
(choosing the $+$ sign for $2m_{jk}\pi$ without loss of generality).

Then for any simple cycle $c= (v_1,v_2,\cdots,v_l,v_1)$ in the cycle basis $B_C$, the winding number is
\begin{align}
    \label{}
    \varpi_c & = \sum_{e \in E} z^c_e \Delta_e (F_e) \\
             & = 2\pi (m_{12}+m_{23},\cdots,m_{l1})  \in \mathbb{Z},
\end{align}
thus completing the first part of the proof.   

        \item [(2) $\implies$ (1)]
          Given, a set of flows satisfying  the dynamic condition 
          \ref{eqn:dc3} and having integral winding numbers, the fixed point 
          $\vec \theta^*$ can be contructed follwoing Eqs \eqref{eqn:gc1} and  
          \eqref{eqn:phase-path}. 
\end{description}
This concludes the proof.  
\end{proof}

\subsection{Geometric frustration} 

The previous reasoning shows that we can face the following
situation: Given an oscillator network characterized by the
frequencies $P_j$ and the capacity matrix $K_{j,\ell}$,
we can find a solution of the dynamical conditions, such that
the flow is conserved at all nodes of the network. Nevertheless,
no fixed point exists as these solutions are incompatible 
to the geometric conditions. In this case we say that phase
locking is inhibited due to \emph{geometric frustration}. We 
summarize this in a formal definition before giving some
examples for the importance of this phenomenon.

\begin{defn}
An oscillator network is said to be \emph{geometrically 
frustrated} if a solution of the dynamic conditions 
(\ref{eqn:dc1}) exits, but all solutions are incompatible
to the geometric conditions (\ref{eqn:gc2}) such 
that no fixed point exists.
\end{defn}

This definition seems unfamiliar at first glance, but is completely compatible 
to the common concept of geometric frustration in condensed matter theory 
\cite{Wann50,Toul80,Moes06}. In that context, a system with multiple state 
variables $(x_1,x_2,\cdots,x_n)$ is called geometrically frustrated \cite{Wolf03}
if there must exist certain pairwise correlations between those variables, and no 
steady state can exist because all these correlations cannot be satisfied 
simultaneously.

To further clarify the relation to condensed matter systems, we consider
a spin lattice system with anti-ferromagnetic interactions
\cite{Wann50} where  the state variables are
the orientation (up or down) of spins.
To minimize the total energy of the system, adjacent spins must be aligned anti-parallel 
introducing perfect pair correlations. Whether this is possible depends
on the geometry or topology of the lattice. It is impossible for triangular
lattices, since two adjacent spins being antiparallel means the third one has 
to be parallel to one of those. Such lattices are thus called frustrated and do not posses a unique 
minimum energy state \cite{Wann50,Toul80}. 
In our case, the correlations are given by equation (\ref{eqn:deltaS}): The 
phases of two adjacent oscillators $j$ and $\ell$ differ by a fixed value defined 
by $\Delta^{\pm}_{j,\ell}$ as given by \eqref{eqn:deltaS}. A fixed point $(\theta_1,\ldots,\theta_N)$ must
satisfy all the correlations, see Eq.~(\ref{eqn:gc1}), otherwise the network is 
said to be frustrated. As before geometric frustration depends crucially on the 
topology of the network as we will show in detail in the following section. 

In general, the problem of geometric frustration can be traced back to the 
fundamental cycles of a lattice or network. In condensed matter physics, 
frustration is classified by the use of plaquette variables, which reveal 
whether a cycle of the lattice contains incompatible correlations \cite{Toul80}. 
In oscillator networks an analog function can be defined for every simple cycle
$c$ starting from the left-hand side of equation (\ref{eqn:gc2}) 
\be
   \Phi_c = \cos \left(\sum_{(j,\ell) \in c} \Delta_{j,\ell} \right) - 1 \, .
   \label{eq:phifun}
\ee
The geometric condition is satisfied for cycle $c$ if $\Phi_c = 0$, whereas 
the cycle is frustrated for $\Phi_c < 0$.

\section{Examples and applications}

In this section we discuss the importance of geometric aspects 
for the fixed points of an oscillator network with different
topologies. In particular, we analyze the number of fixed points
and show that geometric frustration can inhibit phase 
locking, which may lead to counter-intuitive phenomena.

\subsection{Trees do not suffer from frustration.}

By definition, a tree does not contain any cycle such that
the geometric condition (\ref{eqn:gc2}) does not apply. 
Therefore, the calculation of an fixed point  of the 
power grid oscillator model and the Kuramoto model 
as defined by Eq.~(\ref{eqn:def-steady-state})
on a tree reduces to the solution of the dynamic condition 
(\ref{eqn:gc2}), which is a linear set of equations. 
Moreover, we can find a strong result on the 
the number of stable and unstable fixed points
-- see corollary \ref{eqn:corr-tree}.

\subsection{Multiple solutions in cycle}

We now consider the simplest nontrivial topology of a cyclic network with only 
three nodes and three links with equal strength $K$. The 
dynamical condition for the existence of a fixed point then
reads
\be
   K \begin{pmatrix} 0 & 1 & -1 \\ -1 & 0 & 1 \\ 1 & -1 & 0   
   \end{pmatrix}
   \begin{pmatrix} S_{1,2} \\ S_{2,3} \\ S_{3,1} \end{pmatrix}
  = \begin{pmatrix} P_3 \\ P_1 \\ P_2 \end{pmatrix}
    \label{eqn:3cycle-dc} 
\ee
 and $|S_{j,\ell}| \le 1$. In particular for $P_j = 0$ any solution is
 a cycle flow $(S_{1,2}, S_{2,3}, S_{3,1}) = S \times (1,1,1)$.

\begin{figure}[tb]
\centering
\includegraphics[width=\textwidth, angle=0]{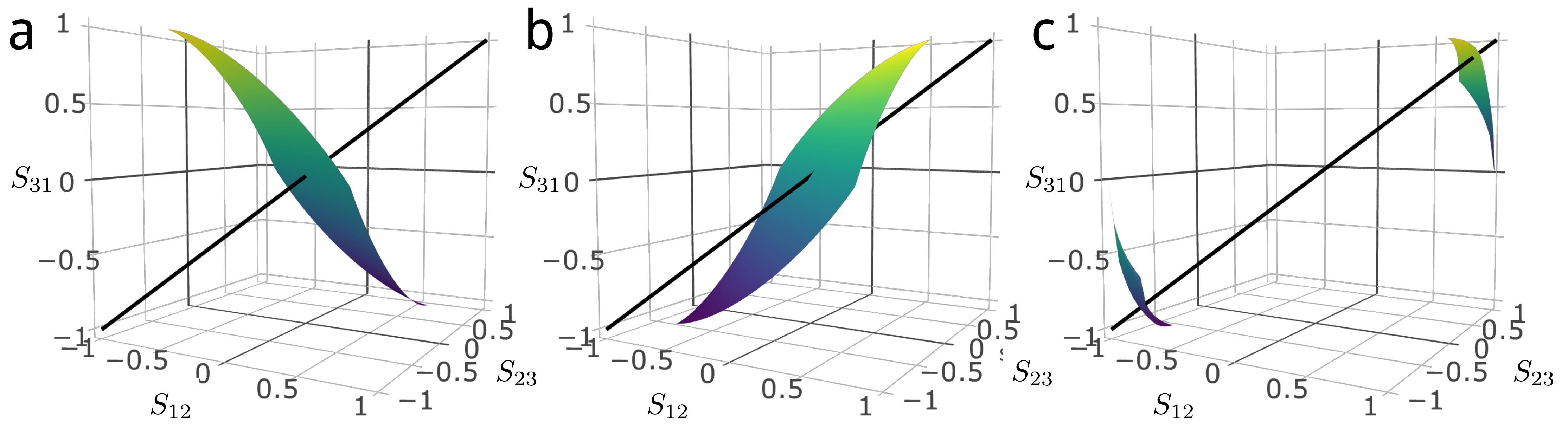}
\caption{\label{fig:cycle3}
Illustration of geometric frustration and multistability in the simplest cyclic network with
3 nodes with $P_j=0$ and three links with equal
strength $K$. Subplots show different branches of \eqref{eqn:3cycle-gc} obtained by 
choosing $+$ or $-$ sign for $\Delta_{12}, \Delta_{23}$ and $\Delta_{31}$. The black
lines denote  the solution space of the dynamical condition 
(\ref{eqn:3cycle-dc}), $S_{1,2} = S_{2,3} = S_{3,1} = S$.   
(a) Branch $(+++)$ with $m=0$.
(b) Branch $(--+)$ with $m=1$. The branches $(+--)$ and $(-+-)$ yield solutions at
$S =(0,0,0)$ in an analogous way.  
(c) Branch $(---)$ with $m=1$ (upper part) and $m=2$ (lower part).
The branches $(++-)$, $(+-+)$ and $(-++)$ do not yield a solution.   
}
\end{figure}

Taking into account that there are two possible solutions for the
phase difference along each edge as per \eqref{eqn:deltaS}), and since in 
order to satisfy the geometric condition \eqref{eqn:gc2}, the sum of
phase differences along the cycle must equal $2m\pi$ for some integer 
$m\in\mathbb{Z}$, we see that all fixed points must satisfy
\begin{align}
    \label{eqn:3cycle-gc}
    \Delta_{12}^{\pm} + \Delta_{23}^{\pm}+ \Delta_{31}^{\pm} & = 2m\pi.  
\end{align}
Taking all combinations of either $\Delta^+$ or $\Delta^-$ and corresponding possible values of 
$m$, we see that there are three intersections corresponding to 
three fixed points. These fixed points are illustrated in Figure 
\ref{fig:cycle3}. This shows that stationary
states are generally not unique, not even for the simplest cycle network.
In the present case only one of the solutions is dynamically stable,
but this is generally not true in larger cycles as we will
show in the following.

\subsection{Frustration induces discreteness.} 
\label{eqn:sec-cycle}

We now extend the above example to a single cycle with an arbitrary number 
of nodes with the same power $P_j \equiv 0$. All links have an equal strength 
$K$ as above.
For the sake of notational convenience we label the nodes as $1,2,\ldots,N$ along the cycle 
and identify the node $1$ with $N+1$ and $0$ with $N$. In order to have a 
non-trivial closed cycle we need $N \ge 3$. 
The dynamic conditions for a fixed points are then given by
\bea
  && F_{j+1,j} = F_{j,j-1}  \equiv F \qquad \text{ for all } j = 1,\ldots,N  \\ 
  && |F| \le K.
\eea
We stress that the dynamic conditions have a continuum 
of solutions, i.e. all values $F$ in the interval $[-K,K]$ are allowed. 

The phase difference along the edges $(j+1,j)$ is
given by equation \eqref{eqn:deltaS}, leaving two possible solutions
$+$ and $-$. Choosing the minus sign for at least one edge $(\ell+1,\ell)$ 
yields $\widetilde K_{\ell+1,\ell}^{\rm red} = - \sqrt{K^2 - F^2} < 0$.
In this case one can show that the Hesse matrix $M$ is not positive
semi-definite such that the fixed point must be unstable. Restricting
ourselves to the dynamically stable states, we find that the phase 
differences are all equal and given by
\be
     \theta_{j+1} - \theta_j = \mbox{arcsin}(F/K).
\ee
The geometric condition now yields
\be
    N \, \mbox{arcsin}(F/K) = 0 \quad (\mbox{mod } 2 \pi),
\ee
which can be satisfied only for certain \emph{discrete} values of $F$. The
geometric condition thus induces a `quantization' of the phase differences
\be
   \theta_{j+1} - \theta_j = \frac{n}{N} 2 \pi, \qquad \mbox{with} \; 
   n \in \left\{ - \left\lfloor \frac{N-1}{4} \right\rfloor, - \left\lfloor \frac{N-1}{4} \right\rfloor+1, 
   \ldots, + \left\lfloor \frac{N-1}{4} \right\rfloor\right\},
\ee
where $\lfloor \cdot \rfloor$ denotes the floor function.   We note that
solutions with $(\theta_{j+1} - \theta_j) = \pm \pi/2$ have jacobian 
eigenvalues $\mu_k = 0$
for all $k \in \left\{1,\ldots,N\right\}$. In this case linear stability 
analysis fails to determine dynamical stability properties (see Khazin und 
Shnol \cite{khazin1991stability} for details). 
For two coupled oscillators it is rather easy to 
see that the fixed point is nonlinearly unstable. In total, we thus find $2 \times \lfloor (N-1)/4 \rfloor +1$ different stable stationary 
states.

This example is very simple but illustrates three important general results.
First, there can be \emph{multiple} stable fixed points in cyclic
networks as previously noticed by \cite{Ocha09,Tayl12,Ioni13,Meht14}. 
This fact is not fully recognized in power engineering, probably 
because most authors in this community concentrate of fully connected 
networks which arise after a Kron reduction \cite{Tayl12,Mott13}.
Second, the oscillator model (\ref{eqn:power}) allows for 
stable fixed points with a persistent current around a cycle.
Interestingly, theses states are phase locked but \emph{not phase
ordered} in the sense that the phase order parameter
\cite{Stro00}
\be
   r e^{i \psi} := \frac{1}{N} \sum_j e^{i \theta_j} 
   \label{eqn:deforder}
\ee
vanishes exactly for $K > 0$.
Third, the geometric condition induces the 
\emph{discreteness} of the phase differences
although the dynamic condition allows for continuous
values of cycle flows.

\subsection{Braess' paradox}
\label{sec:braess}

Here we introduce a special example which illustrates the
paradoxical effects of geometric frustration most clearly.
We consider the oscillator network depicted in Fig.~\ref{fig:braess} (a)
consisting of $N=4$ nodes placed on a cyclic network, where nodes $1$ and $3$ 
have power injection $-P$ and the nodes $2$ and $4$ have power injections $P$.  
In particular, we analyze what happens if the capacity
of the upper edge $(1,2)$ is increased from $K$ to 
$K' = K + \kappa$.

\begin{figure}[tb]
\centering
\includegraphics[width=0.45\columnwidth, angle=0]{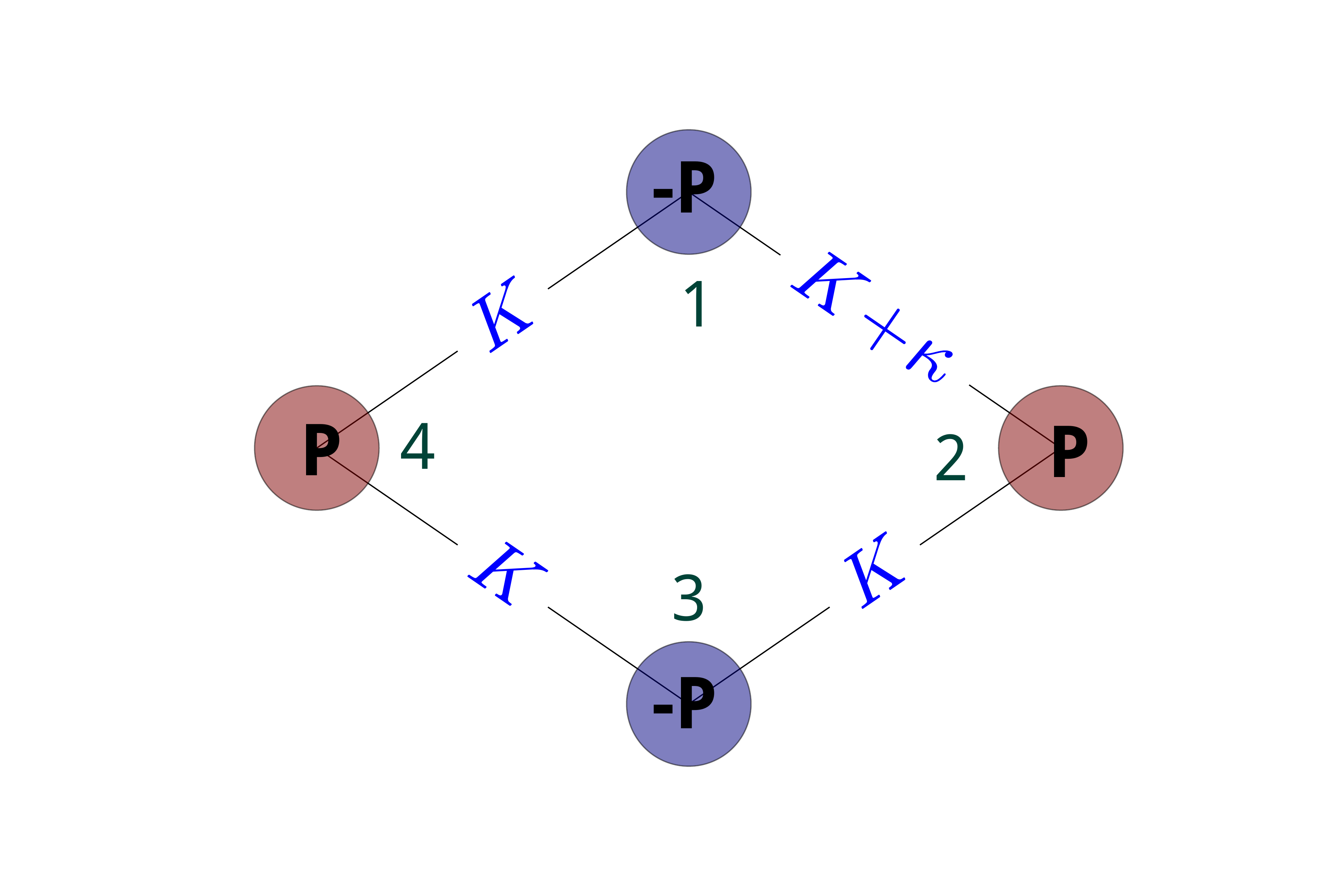}
\includegraphics[width=0.4\columnwidth, angle=0]{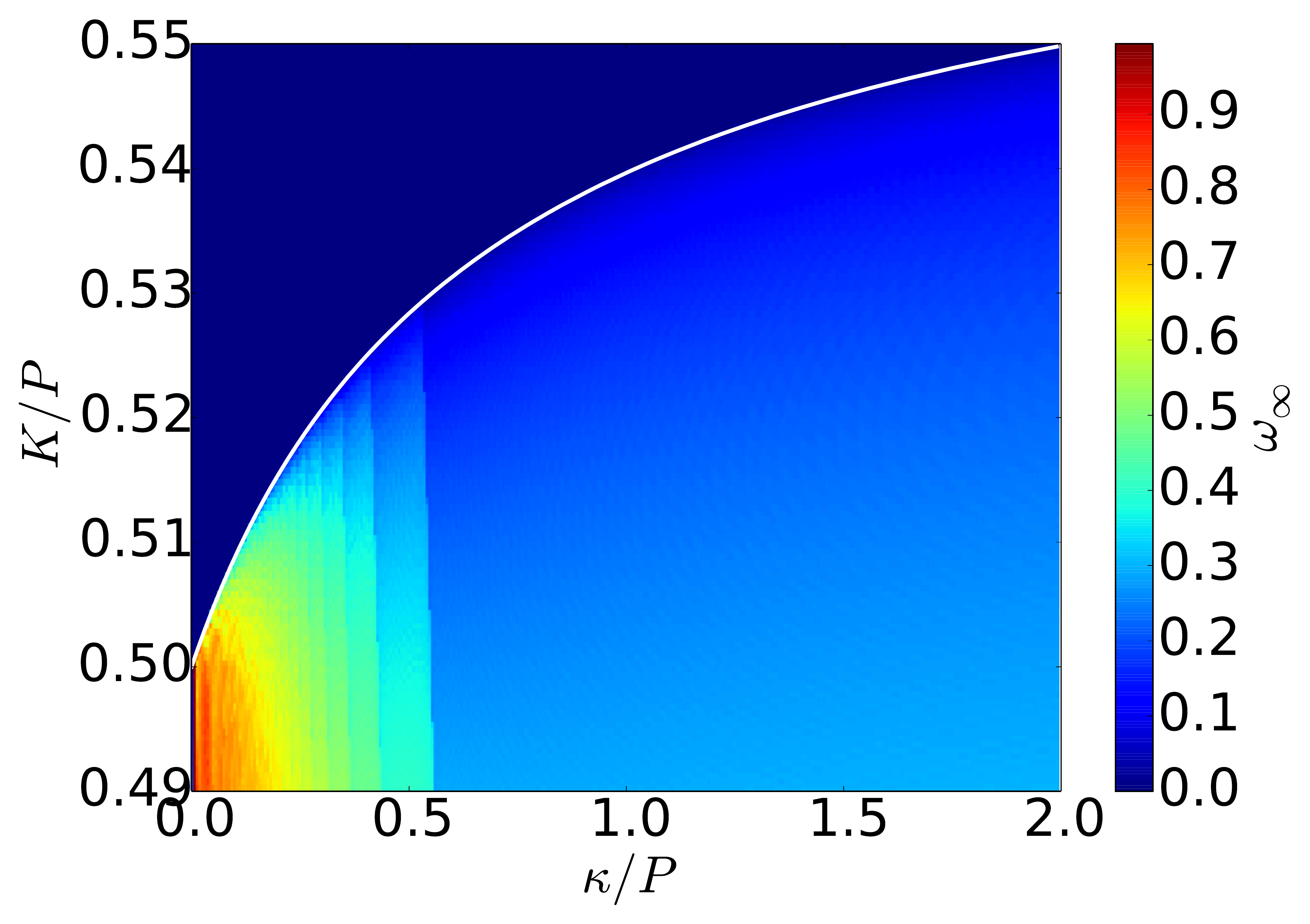}

\caption{\label{fig:braess}
Geometric frustration induces Braess' paradox.
(a) Topology of the network under consideration.
(b) Average phase velocities $\omega_{\infty}$ defined in \eqref{def:omega_inf} for
different values of $K$ and $\kappa$. For fixed points, $\omega_{\infty}=0$.
The white line shows the critical coupling $K_c$. The fixed point can be lost when the
local transmission capacity $\kappa$ 
\emph{increases}.
}
\end{figure}

The dynamic condition for this network reads
\be
   0 = P_j + ( K_{j+1,j} S_{j+1,j}  - K_{j,j-1} S_{j,j-1} ),
   \label{eqn:dc-braess}
\ee
and $|S_{j+1,j}| \le 1$, identifying node $j=5$ with $j=1$.
 For notational convenience,
we define the vector 
\be
   \vec S = (S_{4,1} ,S_{1,2}, S_{2,3}, S_{3,4} ).
\ee
The solutions of the linear system of equations (\ref{eqn:dc-braess}) 
span a one-dimensional affine space parametrized by a real 
number $\epsilon$,
\be
   \vec S = \frac{P}{K} 
      \left( \vec S_a - \epsilon \, \vec S_b \right).
   \label{eqn:SaSb2}
\ee
The vector $\vec S_a = (1,0,K/K',0)$ is the 
inhomogeneous solution of the linear system (\ref{eqn:dc-braess}), 
and the vector $\vec S_b = (1,1,K/K',1)$ is a homogeneous 
solution corresponding to a cycle flow.
Evaluating the condition $|S_{j+1,j}| \le 1$ yields
a necessary condition for the existence of a fixed point
\be
     K \ge P.
\ee

For $\kappa = 0$ this condition is also sufficient for
the existence of a stable fixed point. If the capacity of the
upper link increases, $\kappa > 0$, geometric frustration
inhibits phase locking. A solution of the dynamical conditions
always  exists for   $K \ge P$, but this can become 
incompatible with the geometric condition.   We illustrate this  
in the stability diagram  in Figure \ref{fig:braess} (b). 
A stable fixed point exists only in the parameter region
above the white line. As we see in Figure \ref{fig:braess} (b), the minimum $K$ 
required to maintain steady operation, the \emph{critical coupling} $K_c$, 
increases when $\kappa$ is increased.  

To further characterize the long-time behavior of the oscillator network 
we define $\omega_{\infty}$ as the average phase velocities of all the nodes 
in the limit of large time
\begin{align}
\label{def:omega_inf}
\omega_{\infty}&=\lim_{T\to\infty}\frac{1}{\tau}\int_{T}^{T+\tau} \frac{1}{N}\sum_{j=1}^N \left|\omega_j(t)\right|.
\end{align}
Therefore $\omega_{\infty}$ must be zero for steady operation to take place.
As expected we find $\omega_{\infty}=0$ in the stable parameter region
above the white line $K>K_c$ and $\omega_{\infty}>0$ in the unstable parameter region
below the white line $K<K_c$. Remarkably, $\omega_{\infty}$ is largest for small values of 
$\kappa$ and, of course, $K < K_c(\kappa)$.

This leads to the paradoxical effect that an increase of local 
transmission capacity reduces the ability of the network to support a phase 
locked fixed point.
This behavior can also be seen as an example of 
Braess' paradox \cite{12braess,13nonlocal} which has been first predicted for traffic 
networks \cite{Brae68}.


\section{Multistability and the number of fixed points}

The previous examples show that there can be a large number of stable 
fixed points in a cyclic network. In the following we derive conditions for 
the existence and bounds for the number of stable fixed points depending on the 
network structure.
We start with a deeper analysis of the dynamic condition for arbitrary networks, 
which is a necessary prerequisite for the existence of a stable fixed point.
Then we turn to the geometric condition and derive bounds for
the number for fixed points. The arguments depend heavily on the
network structure such that we will start with trees and simple cycles
before we turn to more complex topologies.

\subsection{The dynamic condition}

We first analyze whether the dynamic condition (\ref{eqn:dc3}) admits a solution.
The problem reduces to the Multi-source multi-sink maximum flow problem, which can
be solved by a variety of different algorithms \cite{nussbaum2010,ford2015flows}.

So let $G=(V,E)$ be a connected graph with $N$ nodes and $L$ edges. Each edge is
assigned a capacity given by $K_1,\ldots,K_L$ and each node has an in- or outflux
given by $P_1,\ldots,P_N$. We define an extended graph $G' = (V',E')$ by adding
two vertices $s$ and $t$ to the vertex set,
\be
  V' = V \cup \{s,t\},
\ee
and adding directed links connecting $s$ ($t$) to all nodes with positive (negative) power injection:
\begin{eqnarray}
    EÕ = E 
    \cup \{(s \rightarrow j) | j \in V, P_j \ge 0 \}
     \cup  \{(j \rightarrow t) | j \in V, P_j  < 0 \}.
\end{eqnarray}
The capacity of the newly added links is infinite. Then one finds the theorem:

\begin{thm}
\label{th:multisource}
A solution of the dynamic condition (\ref{eqn:dc3}) exists if and only if the
value of the maximum $s$-$t$-flow $\mathcal{F}_{st}$ in the network $G'$ is
larger or equal to the cummulated input power 
\be
    \mathcal{F}_{st}  \ge \sum_{j \in V, P_j \ge 0} P_j.
\ee
\end{thm}

Alternatively, a sufficient condition for the existence of a solution can be found from dividing the graph into parts: Let $(V_1,V_2)$ an arbitrary partition of $V$ and $E(V_1,V_2)$ the cutset induced by this partition. Then we define
\be
   \bar P_1 = \sum_{v_j \in V_1} P_{v_j}, \qquad
   \bar P_2 = \sum_{v_j \in V_2} P_{v_j}, \qquad
   \bar K_{12} = \sum_{e \in E(V_1,V_2)} K_e.
\ee
We note that have assumed that $\sum_j P_j = 0$, w.l.o.g, such that we 
always hat $\bar P_1+ \bar P_2 = 0$.

\begin{thm}
\label{th:partition}
If for all partitions $(V_1,V_2)$ we have
\be
   |\bar P_1| = |\bar P_2| \le \bar K_{12}
\ee
then there exists a solution of the dynamic condition (\ref{eqn:dc3}).
\end{thm}

\begin{proof}
The idea is to prove that:\\
$\nexists$ a solution of the dynamic condition (\ref{eqn:dc3a}) and (\ref{eqn:dc3b}).\\
$\Leftrightarrow$ All solutions of (\ref{eqn:dc3a}) violate (\ref{eqn:dc3b}). \\
$\Rightarrow$ $\exists$ a partition $(V_1,V_2)$ with $|\bar P_1| \ge \bar K_{12}$.

Reversing arguments then yields the theorem. It remains to show that the statement 
``$\Rightarrow$'' is true. 

Let $\vec F$ be a solution of (\ref{eqn:dc3a}). According to our assumption the set of 
overloaded edges
\be
   E_{\rm ov} = \{ e \in E | |F_e| > K_e \}
\ee 
is not empty. Now consider one overloaded edge $e_0 = (u,v) \in E_{\rm ov}$. We assume
w.l.o.g that the flow is from $u$ to $v$, i.e.  that $F_{u \rightarrow v} > K_{uv} > 0$. We define
the weighted directed network $\widetilde G(V,\widetilde E)$ with $\widetilde E = E\backslash e_0$ 
and coupling constants
\be
   W_{i \rightarrow j} =  \max\{ 0 , K_{ij} - F_{i \rightarrow j} \}.  
\ee
We determine the maximum flow pattern $\Delta F_e, e \in \widetilde E$ with the
value $\Delta F_{\rm max}$ from $u$ to $v$ in the network $\widetilde G$. According to 
the max-flow min-cut theorem there is a partition $(V_1,V_2)$ with $u \in V_1$ and
$v \in V_2$ and the associated cutset $\widetilde E(V_1,V_2)$ such that
\be
   \Delta F_e = W_e \quad \text{ for all } e \in \widetilde E(V_1,V_2).
\ee
Now consider the flow pattern $\vec F'$ defined by
\begin{align}
   F'_e &= F_e + \Delta F_e \qquad e \in \widetilde E,\\
   F'_{e_0} &=  F_{e_0} - \Delta F_{\max}
\end{align}
This is a new solution of the condition (\ref{eqn:dc3a}). Basically we have rerouted the
maximum possible flow from the edge $e_0 = (u,v)$ to alternative paths from $u$ to
$v$. Furthermore we define the edge set $E(V_1,V_2) = e_0 \cup \widetilde E(V_1,V_2)$, 
which is a cut of the original graph $G$.

We now have to distinguish two cases:

Case 1: The maximum flow value $\Delta F_{\max} < F_{e_0} - K_{e_0}$. Then
the edge $e_0$ is still overloaded, i.e. we have $F'_{e_0} > K_{e_0}$. Summing up 
equations (\ref{eqn:dc3a}) over the nodes in $V_1$ and $V_2$, respectively, yields 
\be
    \bar P_1 = - \bar P_2 = \sum_{e \in E(V_1,V_2)} F'_{e}
\ee
However, we know that $F'_{e_0} > K_{e_0}$ and $F'_{e} = K_{e}$ for
all other $e \in E(V_1,V_2)$ such that
\be
   \sum_{e \in E(V_1,V_2)} F'_{e} > \bar K_{12}
\ee
and the statement ``$\Rightarrow$'' follows.

Case 2: The maximum flow value $\Delta F_{\max} \ge F_{e_0} - K_{e_0}$.
The $e_0$ is no longer overloaded with respect to the flow pattern $\vec F'$. The set
of edges which is still overloaded
\be
   E'_{\rm ov} = \{ e \in E | |F'_e| > K_e \}
\ee
does no longer contain $e_0$, i.e. $E'_{\rm ov} \in E_{\rm ov}\setminus e_0$. However, this
set cannot be empty as we have assumed that there is no solution of (\ref{eqn:dc3a})
satisfying (\ref{eqn:dc3b}). Then we can just restart the procedure, selecting an edge
$e_1 \in E'_{\rm ov}$ and finding a max. flow between its adjacent vectors. Finally
we must arrive at the case 1 for which the statement ``$\Rightarrow$'' follows.

\end{proof}

\subsection{Tree network}

In the previous section we have argued that multistability
arises due to the possibility of cycle flows. In a tree there
are no cycles and thus no multistability and we obtain the
following result.

\begin{corr}
\label{eqn:corr-tree}
In a tree network, there is either no fixed point
or there are $2^{N-1}$ fixed points of which
one is stable and $2^{N-1}-1$ are unstable.
\end{corr}

Whether the fixed points exist or not can then be decided solely on the basis
of the dynamical codition (\ref{eqn:dc1}) resp. using theorem \ref{th:partition}.

\begin{proof}
By definition, a tree has $L=N-1$ edges such that the
space of solutions of the linear system (\ref{eqn:dc2a}) 
has dimension $L-N+1 = 0$. That is, there is either zero
or exactly one unique solution for the flows $F_{j,\ell}$.
In the first case no fixed point exists. In the latter case
there are $2$ possible values for the phase difference 
for each of edge given by equation (\ref{eqn:deltaS}).
Hence there are $2^L = 2^{N-1}$ fixed points.
Choosing the $+$-sign in equation (\ref{eqn:deltaS}) 
yields one stable fixed point as shown in corollary 
\ref{corr:stab-phasediff}.

It remains to show that all other fixed points are unstable.
So consider a fixed point with one edge where the cosine of the phased difference
is smaller than zero. The network is a tree such that it is decomposed
into two parts which are only connected by this edge. We label the
nodes by $1,\ldots,\ell$ in one part and by $\ell+1,\ldots,N$.
in the other part. Then the Hesse matrix $M$ 
(see section \ref{sec:bifurcation})
has the form
\be
    M  = \left( \begin{array}{c|c}
    M_1 & 0     \\   \hline
    0 & M_2 \\ 
   \end{array} \right) +
    \begin{pmatrix}
       \ddots & & & & \\
       & 0 &0 & 0 & 0 &  \\
        & 0 & K^{\rm red}_{\ell,\ell +1} & -K^{\rm red}_{\ell,\ell +1} & 0 & \\
        & 0 & -K^{\rm red}_{\ell,\ell +1} & K^{\rm red}_{\ell,\ell +1} & 0 & \\
       & 0 & 0 & 0 & 0 & \\
      &  &  & & & \ddots \\
  \end{pmatrix}, 
\ee
where $K^{\rm red}_{\ell,\ell +1} < 0$
and $M_1$ and $M_2$ are defined as in equation
(\ref{eq-M}) for the two parts of the network.
Now define the vector
\be
   \vec v = ( \underbrace{1,\ldots,1}_{\ell \, {\rm times}},
                     \underbrace{0,\ldots,0}_{(N-\ell) \, {\rm times}} )^T.
\ee
Due to the structure of the matrix $M_1$ we have 
$M_1 (1,\ldots,1)^T = \vec 0$ such that
\be
   \vec v^T M \vec v =  K^{\rm red}_{\ell,\ell +1} < 0.
\ee
Thus, the Hesse matrix $A$ is not positive semi-definite, 
i.e.~it has at least one negative eigenvalue and the 
fixed point is unstable (cf.~lemma \ref{thm:stability-M}).
\end{proof}

\subsection{Cycle flows and winding vector}

In the following we want to operationalize the theorem (\ref{th:fp-dyn-geo}),
which characterizes fixed points in terms of the flows and winding numbers,
to derive strict bounds for the number of fixed points in a network.
Restricting ourselves to normal operation ($E_- = \emptyset$) and using the
decomposition (\ref{eq:flowdiff_cycle_zero}), the definition of the winding
numbers (\ref{eq:def-winding}) reads
\begin{align}
   \varpi_c &= \frac{1}{2\pi} \sum_{e=1}^{L} z^c_e \arcsin\left( {\frac{F_e}{K_e}} \right) \nn\\
   &=\frac{1}{2\pi} \sum_{e=1}^{L}  z^c_e\arcsin
   \left( \frac{F_e^{(0)}+  \sum_{c' \in B_C} f_{c'} z^{c'}_e }{K_e} \right),
   \label{eq:def-winding-normal}
\end{align}
using equation \eqref{eq:flowdiff_cycle_zero}. The concept of winding numbers is particularly useful when they are unique. If we can find upper and lower bounds for the $\varpi_c$, then we can simply count the number of solutions $\vec \varpi \in \mathbb{Z}^{L-N+1}$ to obtain the number of fixed points. Uniqueness is rigorously shown for planar graphs in the following lemma. 

A graph is called planar if it can be drawn in the plane without any edge
crossings. Such a drawing is called a plane graph or a planar embedding of the
graph and any cycle that surrounds a region without any edges is called a face
of the graph \cite{Dies10}. For the sake of simplicity we adopt the convention
that for plane graphs the cycle basis $B_C$ is built up from the faces in the
following. Notably, many power grids and other supply networks are actually
planar. Crossing of power lines are not forbidden a priori, but are rare.

\begin{lemma}
\label{lem:wind-fp-corr-planar}
For a planar network, let $\vec{\theta}$ and $\vec{\theta'}$ be two fixed points satisfying the ``normal operation'' criterion \eqref{def:normalop}. If $\vec{\varpi}(\vec{\theta})=\vec{\varpi}(\vec{\theta'})$ then both fixed points are the same, i.e. the phases differ only by an additive constant:
\begin{align}
   \vec{\theta}&=\vec{\theta'} + c (1,1,\cdots,1)^T.
\end{align}
In other words, no two different fixed points in planar networks can have identical winding vector.  
\end{lemma}

\begin{proof}
Choose as the cycle basis $B_C$ the faces of the plane embedding. The two fixed points
can differ only via cycle flows such that the flows can be written as
\begin{align}
   \mbox{fixed point} \, \vec  \theta: \quad & F_e = F_e^{(0)} + \sum_{c \in B_C} f_c z^c_e \\
   \mbox{fixed point} \, \vec \theta': \quad & F'_e = F_e^{(0)} + \sum_{c \in B_C} f_c' z^c_e
\end{align}
defining two cycle flow vectors $\vec{f}$ and $\vec{f'}$. We write 
$\vec \varpi(\vec f')$ and $\vec \varpi(\vec f)$ in short-hand notation for 
the corresponding winding vectors. We show that 
$\vec \varpi(\vec f') = \vec \varpi(\vec f)$ implies $\vec f' = \vec f$ and 
thus $\vec F' = \vec F$.   As we are assuming normal operation, we can 
reconstruct the phases via \eqref{eqn:phase-path} and thus find 
$\vec{\theta}=\vec{\theta'} + c (1,1,\cdots,1)^T$ as we need to show.  

So assume that $\vec \varpi(\vec f') = \vec \varpi(\vec f)$  and $ f'_c \neq f_c$
for at least one cycle $c$. We show that this leads to a contradiction such that
the lemma follows. First consider the case that $f_c' - f_c$ is the same for all
cycles, $f'_c - f_c = \Delta f \neq 0$ for all cycles $c \in B_C$. Then choose a
cycle $k$ at the boundary. If $\Delta f > 0$ we find $\varpi_k(\vec f') > \varpi_k(\vec f)$
and if  $\Delta f < 0$ we find $\varpi_k(\vec f') < \varpi_k(\vec f)$. This
contradict the assumption and the lemma follows.

Otherwise, choose a cycle for which the quantity $f_c' - f_c$ is the largest. We can find a cycle $k$ such that
\begin{align}
   f'_k - f_k & \geq f'_\ell - f_\ell  \quad \text{ for all } \; \ell \neq k,  \\
   f'_k - f_k & >     f'_n - f_n \quad \mbox{for at least one cycle} \, n \, \mbox{adjacent to} \, k. 
    \label{eq:planar-cf-greater}
\end{align}
or, equivalently,
\begin{align}
   f'_k-f'_\ell & \geq f_k - f_\ell  \quad \text{ for all } \; \ell \neq k,  \\
   f'_k-f'_n & > f_k - f_n \quad \mbox{for at least one cycle} \, n \, \mbox{adjacent to} \, k. 
\end{align}

We now exploit that any edge belongs to at most two cycles, according to 
Mac Lane's planarity criterion \cite{maclane}. Choosing an edge $e$ which is part 
of both the cycles $k$ and $n$, we have $z^{k}_e z^{k}_e  = 1$ and 
$z^{k}_e z^{n}_e  = -1$.  For all other cycles $\ell\neq k,n$, we have 
$z^{\ell}_e = 0$.  Thus we find  
\begin{align}
    z^{k}_e F_e^{(0)} + \underbrace{z^{k}_e z^{k}_e}_{=+1} f'_k + \underbrace{z^{k}_e z^{n}_e}_{=-1} f'_n +   \sum_{\ell \neq k,n} \underbrace{z^{k}_e z_e^{\ell}}_{=0} f'_\ell  
          &> 
 z^{k}_e F_e^{(0)} + \underbrace{z^{k}_e z^{k}_e}_{=+1} f_k + \underbrace{z^{k}_e z^{n}_e}_{=-1} f_n +   \sum_{\ell \neq k,n} \underbrace{z^{k}_e z_{e}^{\ell}}_{=0} f_\ell       \\
   z^{k}_e F_e^{(0)} + \sum_c z^{k}_e z^{c}_e f'_c & > z^{k}_e F_e^{(0)} + \sum_c z^{k}_e z^{c}_e f_c.  
\end{align}
For every other edge $e'$ in the cycle $k$ we find by the same procedure 
(using \eqref{eq:planar-cf-greater}) that 
\begin{align}
    z^{k}_{e'} F_e'^{(0)} + \sum_c z^{k}_{e'}z^{c}_{e'} f'_c & \ge z^{k}_{e'}F_{e'}^{(0)} + \sum_c z^{k}_{e'}z^{c}_{e'} f_c.  
\end{align}
Substituting these two inequalities in the definition (\ref{eq:def-winding-normal}) and using that
the arcsin is monotonically increasing and point-symmetric about the origin 
such that $\arcsin(z_e^k x) = z_e^k \arcsin (x)$, we find 
\begin{align}        
   \varpi_k(\vec f')    &>      \varpi_k(\vec f).
\end{align}
This contradicts our contrary assumption, which concludes the proof.  
\end{proof}

\subsection{Simple cycles}
\label{sec:cycles}

For networks containing a single cycle, tight upper and lower bounds
can be obtained for the number of fixed points satisfying
$\cos(\theta_i^*  - \theta_j^*) > 0$ for all edges $(i,j)$. These states 
correspond to the normal operation of a power grid and are guaranteed 
to be stable by corollary \ref{corr:stab-phasediff}. Other stable steady 
states can exist in particular at the border of the stable parameter 
region \cite{epjst14}. 
We label the nodes as $1,2,\ldots,N$ along the cycle, fixing the
direction of counting in the counter-clockwise direction and identify the 
node $1$ with $N+1$ and $0$ with $N$. Likewise we fix the orientation 
of the edges $e \in \left\{1,\ldots,L\right\}$ such that $F_e>0$ describes a 
counter-clockwise flow and $F_e<0$ a clockwise flow.

We will first calculate the exact number of fixed points counting the number 
of different allowed winding numbers. 
However, this result depends on one particular solution of the dynamic conditions
(\ref{eqn:dc1}), thereby limiting its applicability. 
We therefore also derive lower and upper bounds for the number of 
fixed points in terms of a few simple characteristics of the grid, in 
particular the maximum partial net power.   These bounds do not depend on any 
particular solution of the dynamical condition.

\begin{remark}
\label{rem:ring-windnum}
For any ring network $\R_N$ with $N$ nodes, the cycle flow vector defined in
\eqref{lem:cf-fp-corr} and the winding vector 
defined in \eqref{def:wvec} naturally reduce into single numbers. We refer to 
them as \emph{cycle flow} 
$f_c$ and \emph{winding number} $\varpi_c$, following \cite{Ocha09}. These two 
quantities will be crucial in establishing the results in the rest of this 
section.  
\end{remark}

\begin{thm}
\label{lemma-cycle}
For a ring network $\R_N$, the number of normal operation fixed point  (denoted by $\N$) 
is given by 
\begin{align}   
   \N =  \ceil{\frac{1}{2\pi} \sum_j\arcsin\left(\frac{F_{j+1,j}^{(0)}+f^{\rm max}_c}{K_{j+1,j}}\right) }
    - \floor{\frac{1}{2\pi}\sum_j\arcsin\left(\frac{F_{j+1,j}^{(0)}+f^{\rm min}_c}{K_{j+1,j}}\right) }
     -1
\end{align}
where $\floor{\cdot}$ denotes the floor function and $\ceil{\cdot}$ denotes
the ceiling function. 
$F_{ij}^{(0)}$ is one particular solution to the dynamic condition (\ref{eqn:dc1})
and
\begin{align}
   \label{eqn:def-Fmaxmin}
   f^{\rm max}_c &= \min_j \left( K _{j+1,j} - F_{j+1,j}^{(0)} \right), \nn \\
   f^{\rm min}_c &= \max_j  \left( - K _{j+1,j} - F_{j+1,j}^{(0)} \right).
\end{align}
\end{thm}

\begin{proof}
Suppose we have one fixed point $\vec{\theta_0}$
with the flows $F_{ij}^{(0)}$ and analyze (as per throrem \ref{lem:wind-fp-corr-planar})
which cycle flow values $f_c$ lead to different valid fixed points. 
First, the cycle flow is bounded both above and below
since the flow $F_{j,j+1}$ along each edge cannot exceed in 
absolute value the capacity $K_{j,j+1}$:
\begin{align}
  \label{eq:fp_limits}
  f^{\rm min}_c & < f_c < f^{\rm max}_c\\
  f^{\rm max}_c&=\min_{j} \left(K_{j,j+1}-F^{(0)}_{j,j+1} \right)\\
  f^{\rm min}_c&=\max_{j} \left(-K_{j,j+1}+F^{(0)}_{j,j+1} \right).
\end{align}
We emphasize that $f_c$ cannot be equal to $ f^{\rm max}_c$
or  $f^{\rm min}_c$, because otherwise one edge would be fully loaded with 
$\cos(\theta_i - \theta_j) = 0$, contradicting our assumption.

Second, all fixed points have to satisfy the geometric condition (cf. theorem \ref{th:fp-dyn-geo})
\be
\varpi(f_c) \in \mathbb{Z}.
\ee
Since we restrict ourselves to normal operation, the winding number for a single cycle reads 
\be
   \varpi(f_c) = \frac{1}{2\pi} \sum_j\arcsin\left(\frac{F_{j+1,j}^{(0)}+f_c}{K_{j+1,j}}\right).
\ee
Using the bound for the cycle flow strength (\ref{eq:fp_limits}), and the fact that the arcsin is a monotonically increasing function, we find that the winding number is also bounded by 
\begin{align}
   \varpi(f^{\rm min}_c) \le \varpi \le \varpi(f^{\rm max}_c)
  \label{eq:lim-omega}
\end{align}
As the winding numbers are unique (see lemma \ref{lem:wind-fp-corr-planar}), the distinct fixed points correspond to the following values of the winding number: 
\be
    \varpi_{\rm fixed point} = \floor{ \varpi(f^{\rm min}_c) } + 1,
           \floor{ \varpi(f^{\rm min}_c)  } + 2, \ldots,
           \ceil{ \varpi(f^{\rm max}_c)  } -1 .
\ee
Counting these values and inserting the values for $f^{\rm min}_c$
and $f^{\rm max}_c$ then yields the number of fixed points $\N$.
\end{proof}

For practical applications it is desirable to determine the number of fixed points
from the properties of the network alone, without referring to a particular solution
$F^{(0)}$. To obtain suitable bounds for the number of fixed points $\N$, we 
first define some properties which characterize the network.

\begin{figure}[!htp]
\begin{center}
\includegraphics[width=0.6\columnwidth]{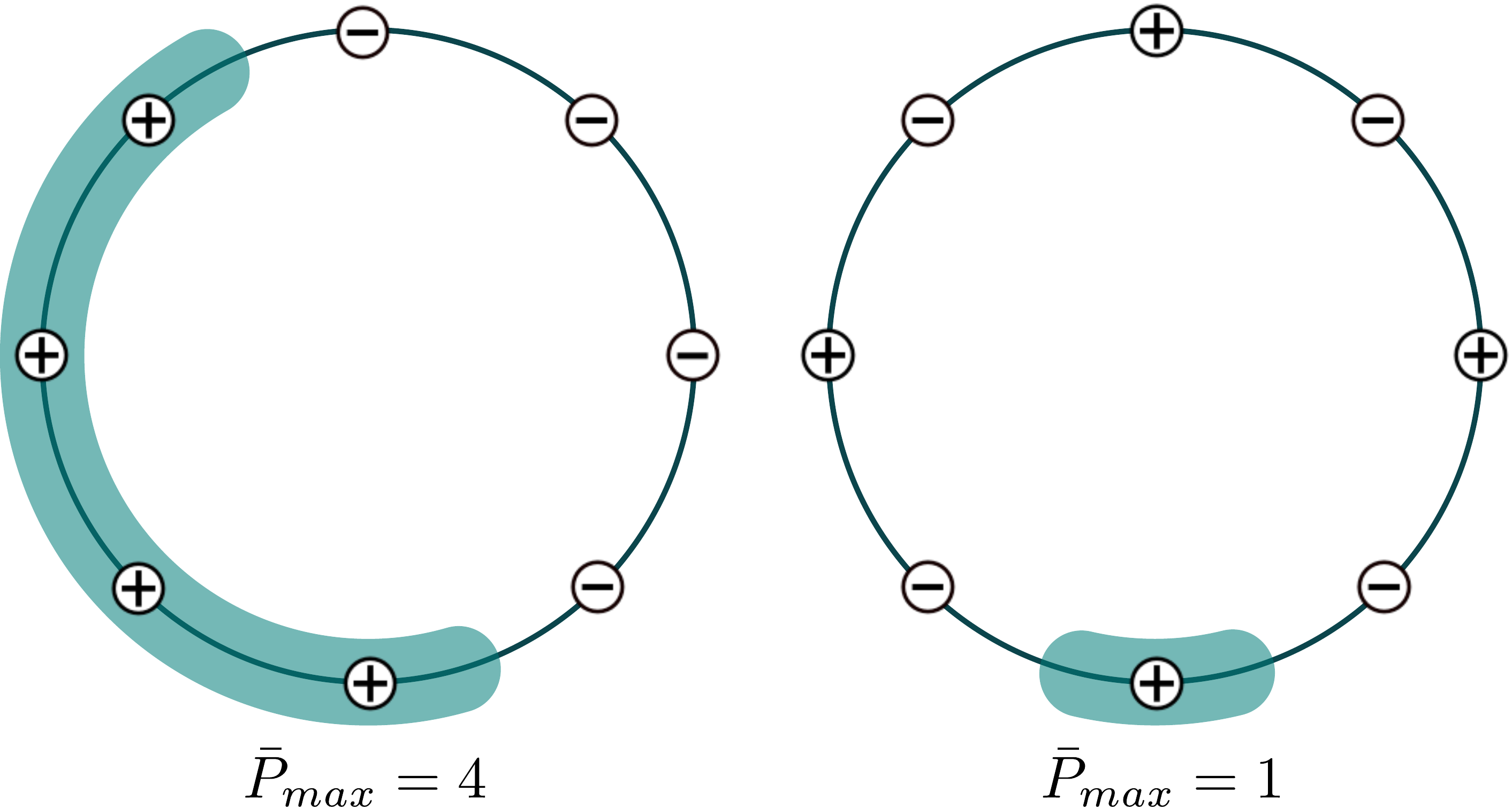}
\end{center}
\caption{
\label{fig:pmax}
  The maximum partial net power $\bar{P}_{\max}$ in different ring networks.
}
\end{figure}

\begin{defn} 
For a ring network $\R_N$ with $N\in \mathbb{N}$ nodes indexed by $1,2,\ldots,N$
along the cycle, a \textbf{fragment} $\F_{i,j}$ is defined as the path starting at node $i$ 
and ending at node $j$. For any fragment $\F_{i,j}$, the \textbf{partial net power} 
$\bar{P}_{ij}$ is defined as
\be
   \bar{P}_{ij}=\sum_{k=i}^j P_k.  
\ee
and the \textbf{maximal partial net power} is defined as
\begin{align}
  \label{def:mpnp}
   \bar{P}_{\max}=\max_{i,j} \bar{P}_{i,j}.
\end{align}
This concept is illustrated in Fig.~\ref{fig:pmax}
Furthermore we define the maximum and minimum transmission capacities
\be
    K_{\rm max} = \max_{j} K _{j+1,j} 
     \quad \text{and} \quad
    K_{\rm min} = \min_{j} K _{j+1,j}.
\ee
\end{defn}

\begin{lemma}
\label{lem:flow-power}
For any ring fragment $\F_{i,j}$, the partial net power is equal to the net outwards flow:
\begin{align}
  \label{eq:glow-power}
  \bar{P}_{ij}=F_{j+1,j}-F_{i-1,i}
\end{align}
and
\begin{align}
    \label{eq:pmax}
   \bar{P}_{\max} = \max_j F_{j+1,j} - \min_{i} F_{i-1,i} \, .
\end{align} 
\end{lemma}

\begin{corr}
\label{corr:hom-ring}
For a ring network $\R_N$, the number of normal operation fixed points (denoted by $\N$)
is bounded from above and below by
\begin{align}
   \label{eq:Nring1}
   0  \leq \N \leq 2 \floor{ \frac{N}{4} }  +1 
\end{align}
and by
\begin{align}
\ceil{ \frac{N}{4} \, \frac{2 K_{\rm max}- \bar 
   P_{\rm max}}{K_{\rm min}}} & \geq \N \geq \ceil{ \frac{N}{2\pi} 
           \frac{2 K_{\rm min} - \bar P_{\rm max}}{K_{\rm max}} } - 1
     \label{eqn:bounds:homring2}      
\end{align}
\end{corr}

\begin{proof}
According to lemma \ref{lemma-cycle} the number of fixed points $\N$ 
is given by
\begin{align}
    \N = \ceil{\varpi(f^{\rm cycle}_{\max\phantom{i}})} -\floor{\varpi(f^{\rm min}_c)} -1  .
\end{align}  
We make use of the fact that the arcsin function is bounded, 
$\arcsin{(x)} \in [-\pi/2,+\pi/2]$, such that   
\begin{align}
    \label{eqn:omega_extremes}
    \varpi(f^{\rm max}_c) &= 
    \frac{1}{2\pi}  \sum_{j=1}^{N} \arcsin \left(\frac{F_{j+1,j}^{(0)}+f^{\rm max}_c}{K_{j+1,j}}\right)
     < \frac{N}{4}   \nn \\
     \varpi (f^{\rm min}_c) &= 
      \frac{1}{2\pi}  \sum_{j=1}^{N} \arcsin \left(\frac{F_{j+1,j}^{(0)}+f^{\rm min}_c}{K_{j+1,j}}\right)
     > - \frac{N}{4}
\end{align}
This proves the first part \eqref{eq:Nring1} of the corollary. To proof the second part, we start from 
\begin{align}
   \label{eq:nf-bounds}
   \ceil{\varpi(f^{\rm max}_c) - \varpi(f^{\rm min}_c)} - 1
      &\leq \N \leq 
    \ceil{\varpi(f^{\rm max}_c) - \varpi(f^{\rm min}_c)} 
\end{align}
Now one can obtain upper and lower bounds for all terms in the
sum using the trigonometric relation 
\begin{align}
   \label{eqn:xy-arcsin}
   x-y\leq \arcsin(x)-\arcsin(y)\leq \frac{\pi}{2}(x-y)
\end{align}
which holds for all $-1\leq y \leq x \leq 1$. This yields 
\begin{align}
\label{eq:q-bounds}
   \frac{1}{2\pi} \sum_j\frac{\Delta f_c}{K_{j+1,j}} &\leq 
    \varpi(f^{\rm max}_c) - \varpi(f^{\rm min}_c) \leq 
   \frac{1}{4} \sum_j\frac{\Delta f_c}{K_{j+1,j}},
\end{align}
where we define $\Delta f_c = f^{\rm max}_c - f^{\rm min}_c$.
Furthermore, this quantity can be bounded as
\begin{align}
  \Delta f_c &= \min_j \left( K _{j+1,j} - F_{j+1,j}^{(0)} \right) -
                     \max_j  \left( - K _{j+1,j} - F_{j+1,j}^{(0)} \right) \nn \\
               & \geq 2 K_{\rm min} +  \min_j \left( - F_{j+1,j}^{(0)} \right)
                      - \max_j  \left( - F_{j+1,j}^{(0)} \right) \nn \\
               & = 2 K_{\rm min} +  \max_j \left( F_{j,j+1}^{(0)} \right)
                      - \max_j  \left( - F_{j+1,j}^{(0)} \right) \nn \\       
                      & = 2 K_{\rm min}- \bar P_{\rm max}  \quad\text{(using \eqref{eq:Nring1})},
\end{align}
such that the fraction in equation (\ref{eq:q-bounds}) becomes
\begin{align}
     \frac{\Delta f_c}{K_{j+1,j}} \geq  \frac{2 K_{\rm min}- \bar P_{\rm max}}{K_{\rm max}} .
\end{align}
In a similar way we find  
\begin{align}
\Delta f_c \leq 2 K_{\rm max}- \bar P_{\rm max}.
\end{align}
Substituting these bounds into equation (\ref{eq:q-bounds}) yields
\begin{align}
  \varpi(f^{\rm max}_c) - \varpi(f^{\rm min}_c)
      & \geq \frac{N}{2\pi} \, \frac{2 K_{\rm min}- \bar P_{\rm max}}{K_{\rm max}} \nn \\
  \varpi(f^{\rm max}_c) - \varpi(f^{\rm min}_c)
       & \leq \frac{N}{4} \, \frac{2 K_{\rm max}- \bar P_{\rm max}}{K_{\rm min}} ,
\end{align}
which combined with \eqref{eq:nf-bounds} completes the proof.

\end{proof}

\begin{corr}
\label{cor:high-k-nf}
\label{corr:nomult-3}
For homogeneous rings $\R_N$, i.e.~$K_{i,i+1}=K$, equation (\ref{eqn:bounds:homring2})      
simplifies to
\begin{align}
   \ceil{ \frac{N}{\pi}-\frac{N\bar{P}_{\max}}{2K\pi} } -1  \leq \N \leq 
      \ceil{ \frac{N}{2} - \frac{N\bar{P}_{\max}}{4K} }.
\end{align} 

In particular, ring networks $\R_N$ with $N \le 4$ do not have multiple stable fixed points.
Ring network $\R_{N}$ with $N \ge 7$ nodes 
will have multiple stable fixed points $(\N \ge 2)$ if 
\begin{align}
   \label{crit:mult-fp}
   \bar{P}_{\max} < 2 K_{\rm min} - \frac{4 \pi }{N} K_{\rm max} \, .
\end{align}
\end{corr}
These relations can be proven by simply evaluating the bounds in corollary  
\ref{corr:hom-ring}.

\begin{corr}
\label{corr:fp-norm}
As $K$ is decreased, the fixed points with higher magnitude of the \emph{winding number} 
\eqref{def:wvec} will vanish first. From another point of view, the fixed points with 
largest infinity norm of the flows
\[
||\vec{F}||_{\infty} := \max_j\left|F_{j,j+1}\right|
\]
will be the first ones to vanish.  
\end{corr} 

\begin{proof}
We can see from \eqref{eqn:omega_extremes} that 
$\varpi(f^{\rm max}_c)$ and $\varpi(f^{\rm min}_c)$ are 
both monotonically increasing functions of $f^{\rm max}_c$ and 
$f^{\rm min}_c$, respectively.  
According to \eqref{eq:fp_limits}, when $K$ is decreased, 
$f^{\rm max}_c$ decreases and $f^{\rm min}_c$ increases.  
The corollary follows.  

\end{proof}

\begin{figure}[htb]
\begin{center}
\includegraphics[width=\columnwidth]{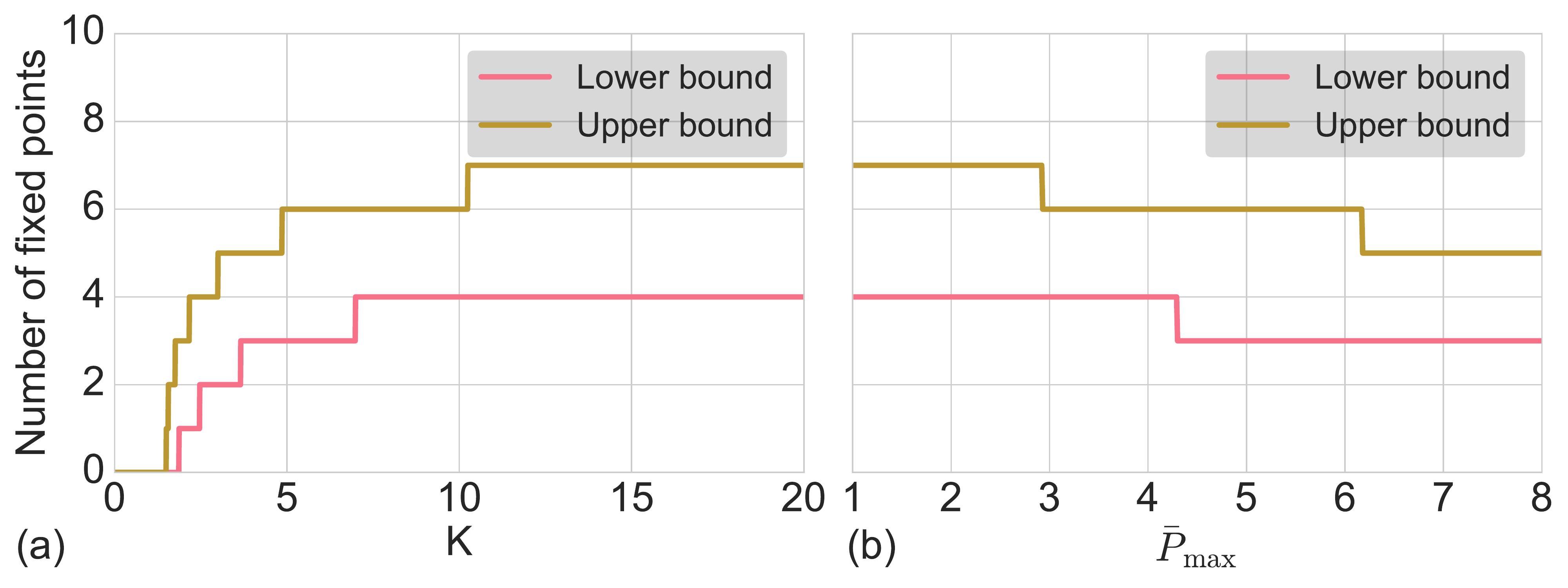}
\caption{
Upper and lower bounds for the number of fixed points $\N$ 
for a sample 16 element ring as a function of (a) $K=K_{j,j+1}$ for all 
$1\leq j \leq 16$ at $\bar{P}_{\max}=3$, and (b) $\bar{P}_{\max}$ at $K=10$.  }
\label{fig:scaling-nf}
\end{center}
\end{figure}

We illustrate how the bounds scale with the connectivity $K$  and $\bar{P}_{\max}$ for a sample ring 
of size $N=16$ in  Fig \ref{fig:scaling-nf}. We see in Fig \ref{fig:scaling-nf} 
(a) that increasing $K$ results in more stable fixed points.  Whereas Fig 
\ref{fig:scaling-nf} (b)  demonstrates that if the power generators 
$(P_j\ge  0)$ are clustered together, then the system has less fixed points, 
as opposed to the case where they are more distributed.

\subsection{Complex networks}
\label{subsec:complexnet}
Obtaining bounds for the number of fixed points is hard in general, as we 
cannot simply decompose a network into single cycles, unless no two cycles of 
a network share an edge.  The diffculty arises 
because cycle flows in two faces sharing one or more edges can cancel or 
enhance each other.  That is why one cannot simply multiply the bounds for 
number of fixed points for each cycle to obtain a bound for the total number 
of fixed points for a network.  We demonstrate this using two examples.

\subsubsection{Two cycle flows destroying each other}
First, we show that even if all single 
cycles support (multiple) fixed points in case there are isolated, the full 
network may not have a single fixed point at all. This is illustrated in Figure 
\ref{fig:lowerbound-nosol} for a network consisting of just two cycles. The 
network motifs shown in panels (a) and (b) have 3 and 1 stable fixed point, 
respectively, whereas the full network shown in panel (c) does not have a 
stable fixed point.
The isolated cycle 2, i.e. the network shown in Figure \ref{fig:lowerbound-nosol} (b), 
has a stable fixed point but two edges are heavily loaded such that there is nearly no 
security margin and no available capacity for cycle flows. Fusing the two cycles as in 
Figure \ref{fig:lowerbound-nosol} (c) disturbs the geometric condition for both cycles.
To restore the geometric condition $\vec \varpi \in \mathbb{Z}^2$ we would have to add 
some cycle flows. But this is impossible in the cycle 2 such that there is no stable fixed 
point in the full network.

\begin{figure}[tb]
\centering
\includegraphics[height=0.15\textheight]{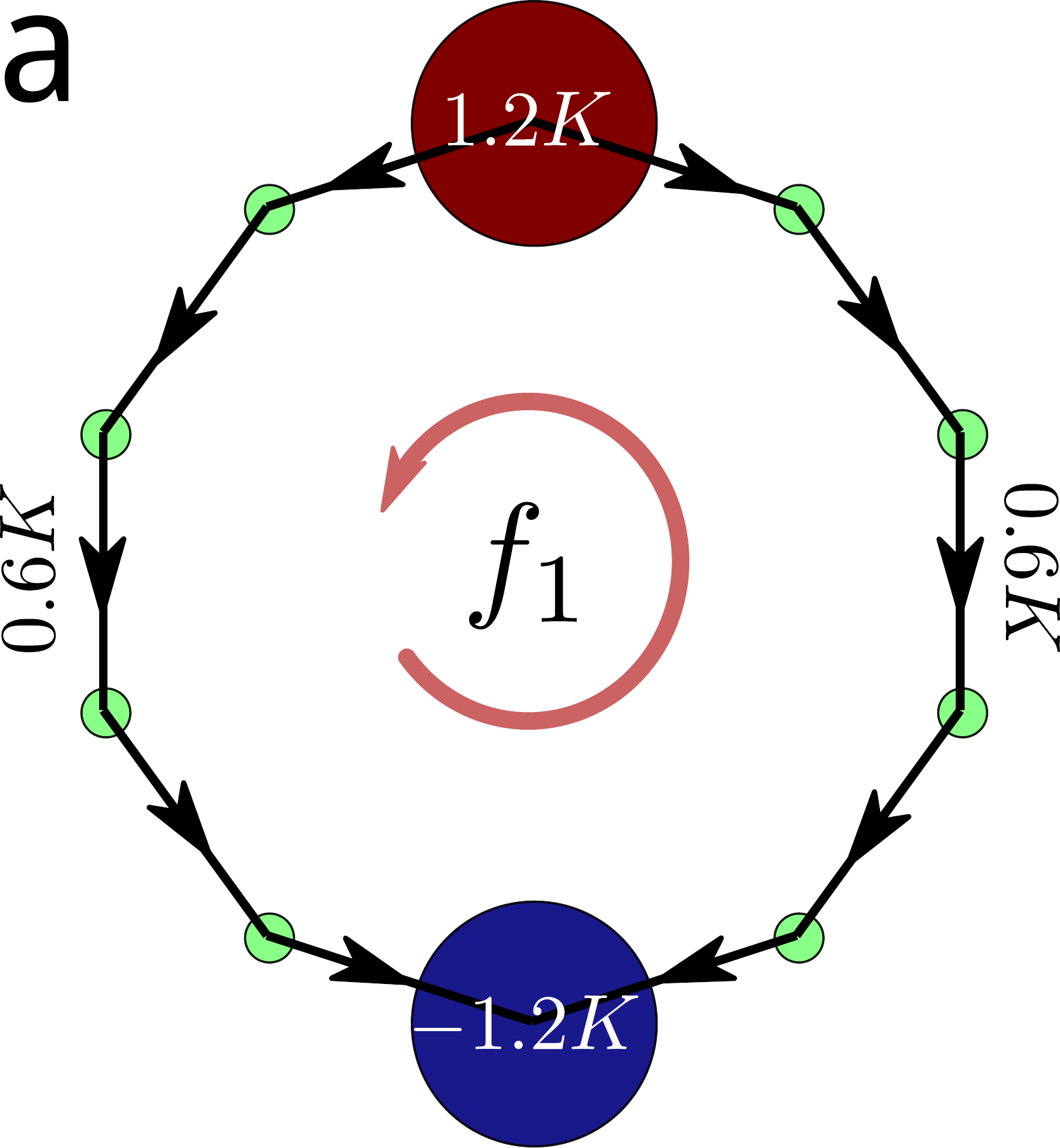}
\hspace{5mm}
\includegraphics[height=0.15\textheight]{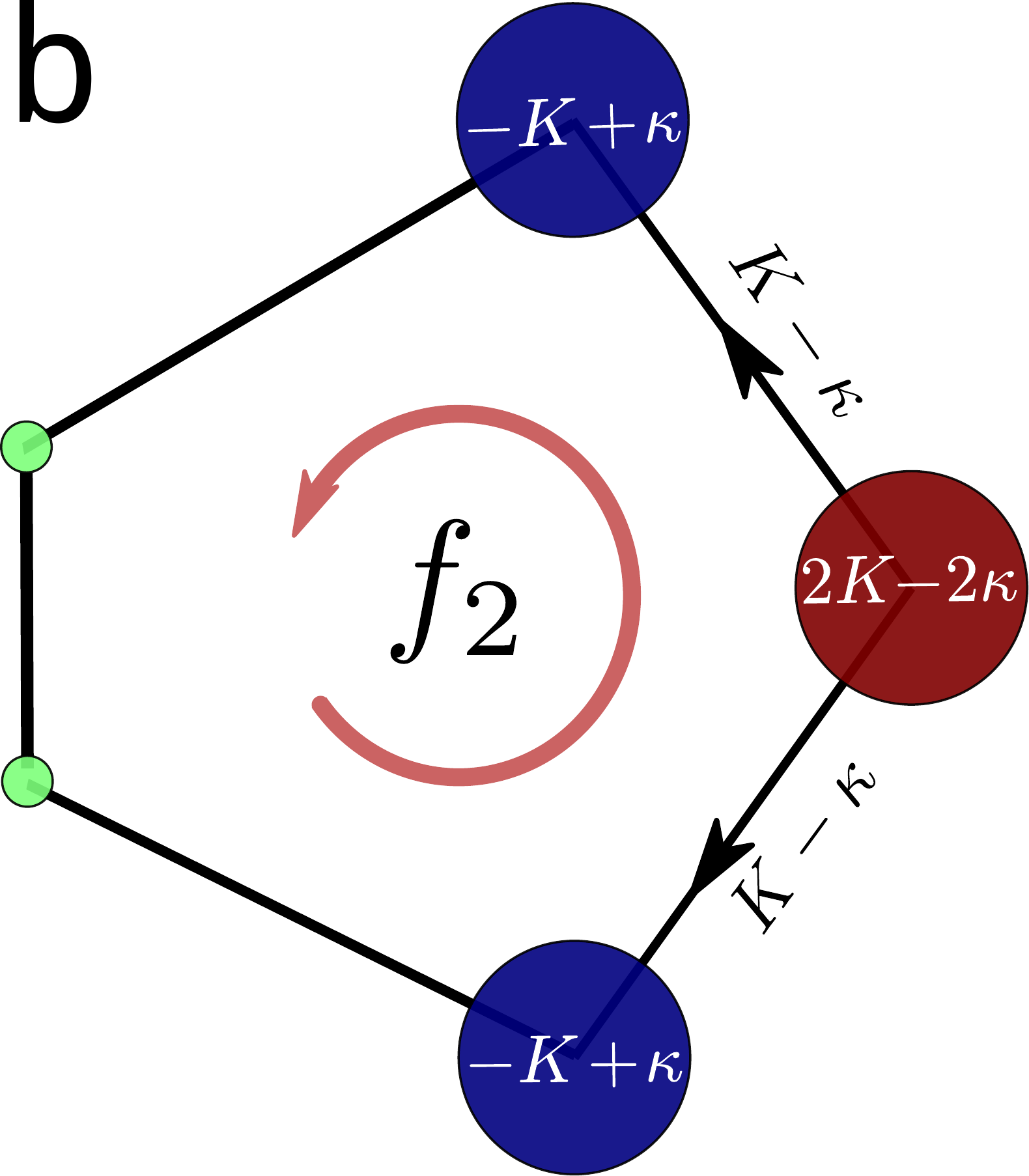}
\hspace{5mm}
\includegraphics[height=0.15\textheight]{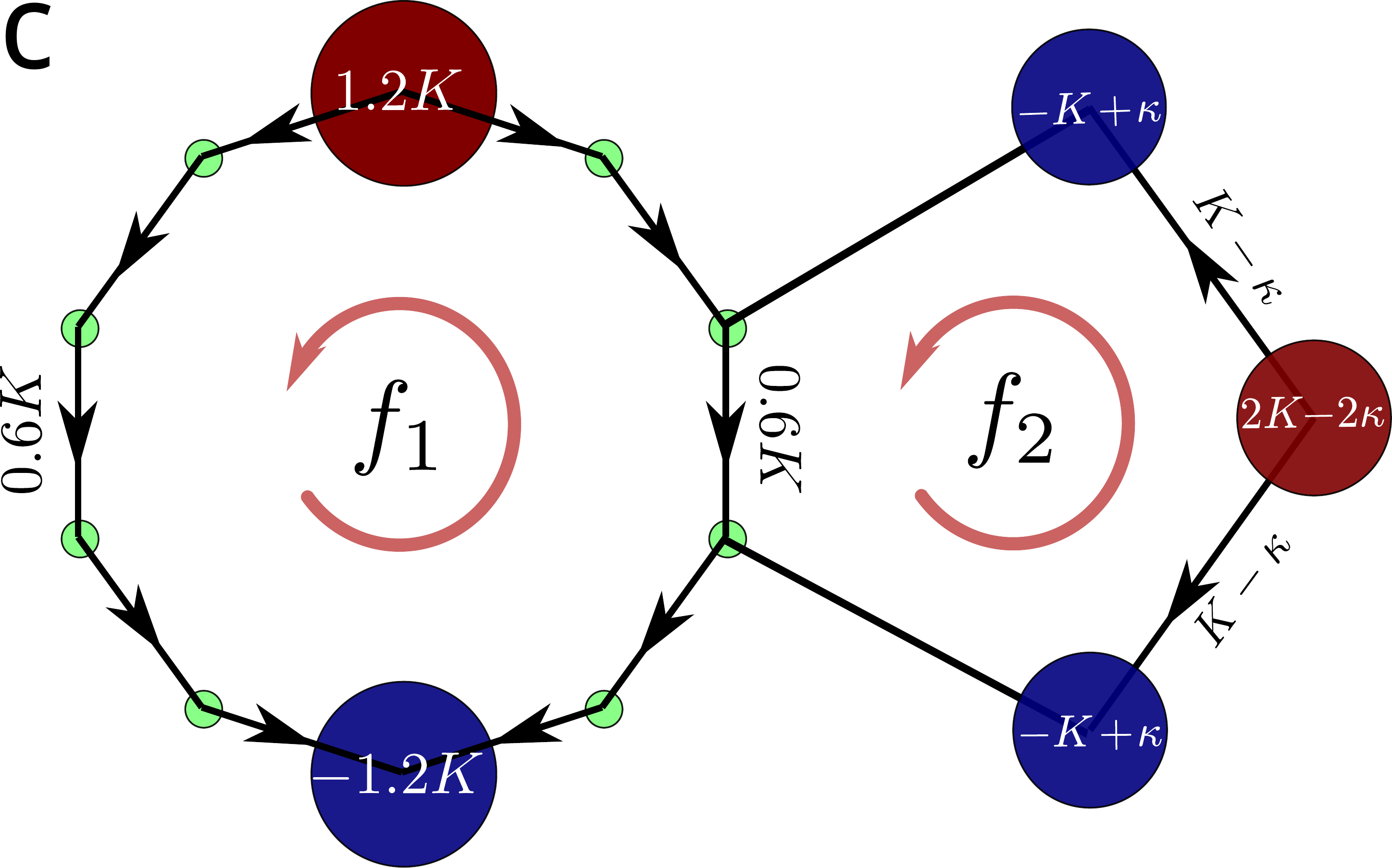}
\caption{
\label{fig:lowerbound-nosol}
Difficulties in finding bound for the number of stable fixed points in complex network.
The network motifs shown in (a) and (b) have 3 and 1 stable fixed point, respectively, whereas 
the fused network shown in (c) has no stable fixed point at all. The power 
injections $P_j$ are given in the nodes.  All edges have
transmission capacity $K$.  
}
\end{figure}

\subsubsection{Two cycle flows getting created}

\begin{figure}[!htb]
\centering
\includegraphics[width=0.45\textwidth]{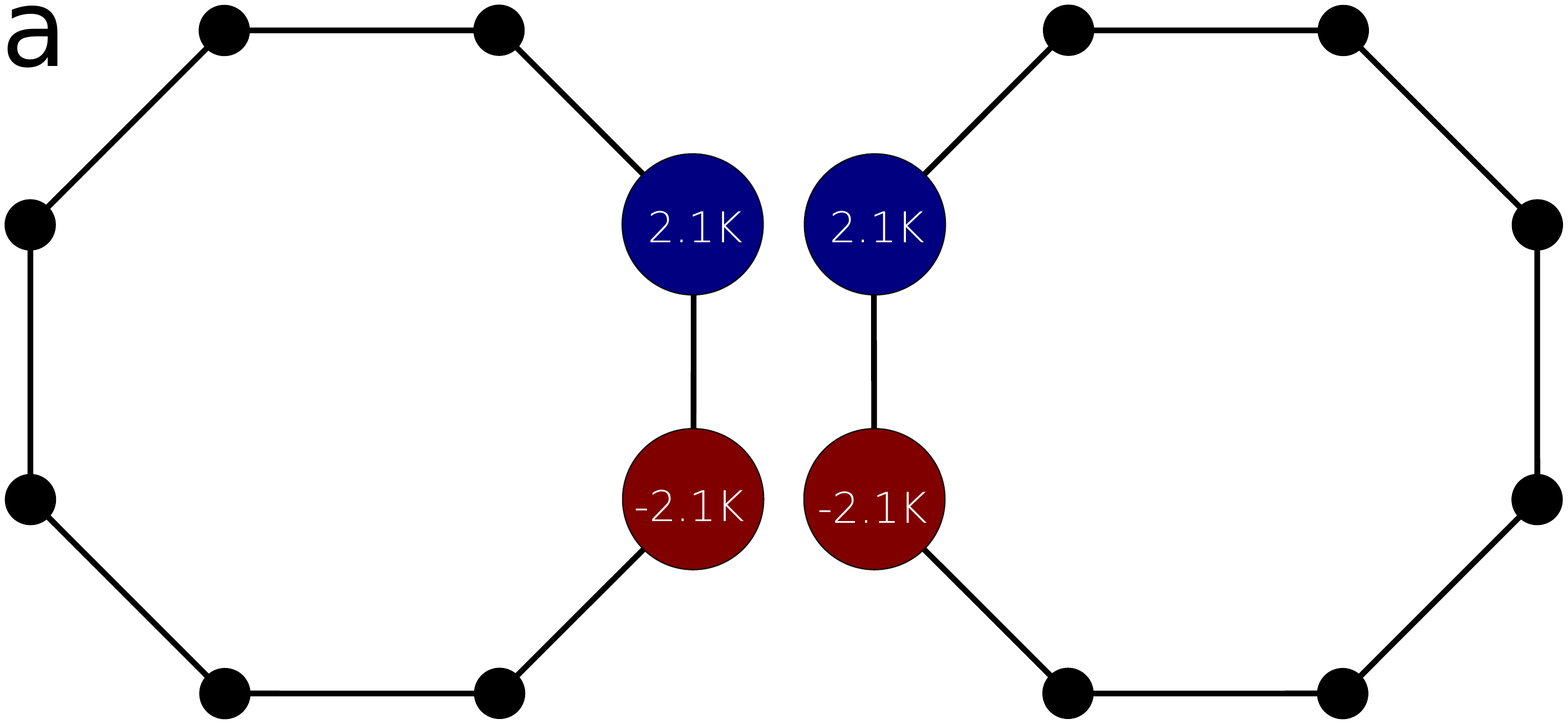}
\hspace{5mm}
\includegraphics[width=0.45\textwidth]{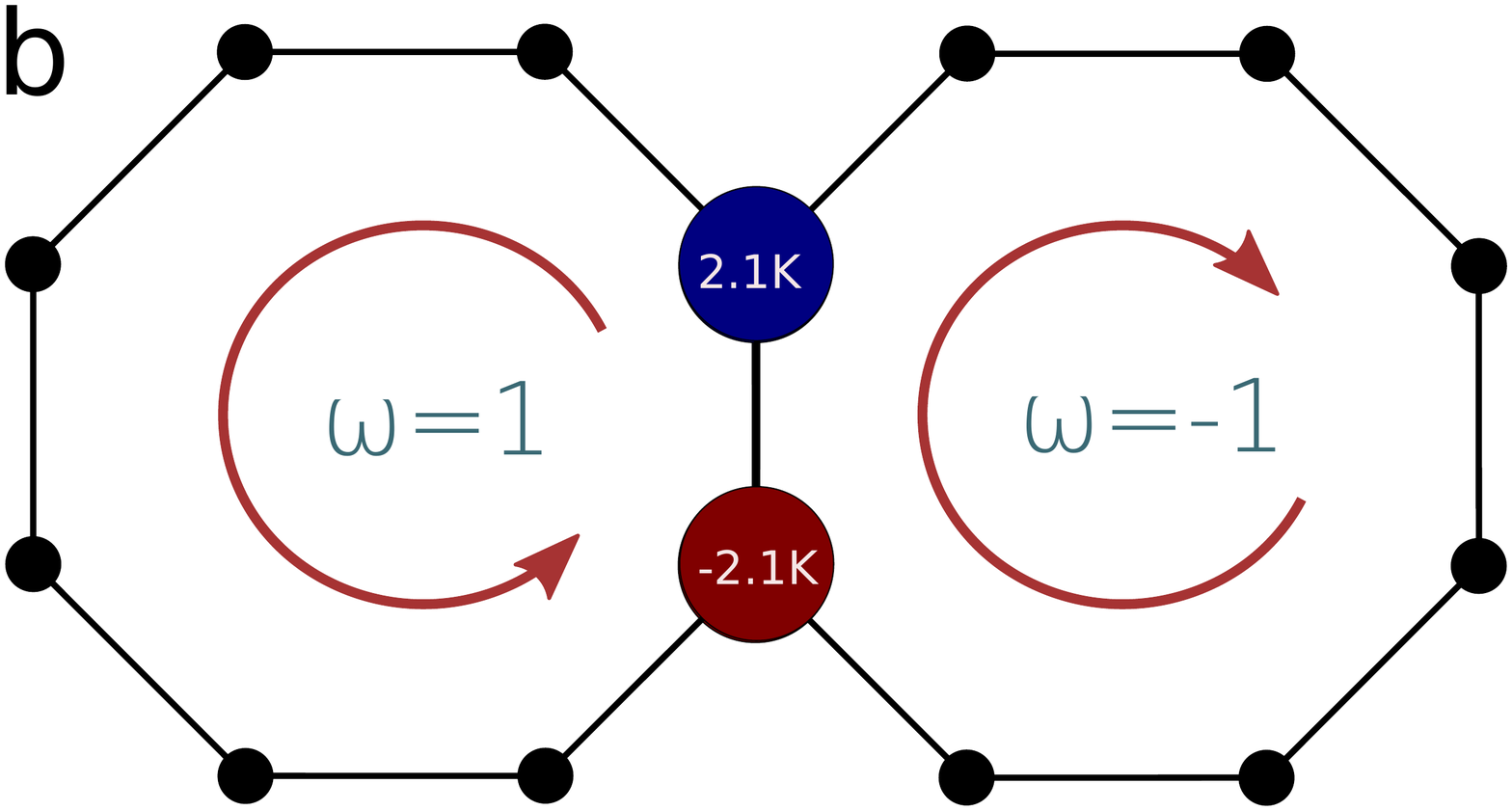}
\caption{\label{fig:cf_fusion_create}
Two ring networks, each with no fixed point, when merged by an edge, gains a 
fixed point.
}
\end{figure}

So we have seen that getting a \emph{lower bound} for number of fixed points 
of a general network is hard, as multiplying lower bounds for each cycle in a 
cycle basis does not yield a valid lower bound.  Next we will show why 
obtaining a good \emph{upper bound} is also hard.  

Consider any of the two identical single loop networks in Figure \ref{fig:cf_fusion_create}. It consists of one 
generator and one consumer, generating and consuming $2.1K$ power 
respectively.  Each edge has capacity $K$.   None of the two single loop networks have any fixed point: simply because the network does 
not have enough capacity to transport the $2.1K$ 
amount of power from the generator to teh consumer.  However, when those two are 
fused together, two cycle flows emerge and a stable fixed point with winding vector 
$\vec{\omega}=(1,-1)$ comes into existence.  This should not come as a 
surprise: fusing two cycles in this case ended up with an alternate pathway 
for the powerflow being created. 

\subsection{Planar networks}
Although obtaining estimates for number of fixed points for general topologies is 
quite difficult, we now show that for planar topologies, it is possible to 
derive some analytical insights.
\subsubsection{Upper bound}
\begin{thm}
Consider a finite planar network. Choose the faces of the graph as the cycle 
basis $B_C$.   Then the number of normal operation fixed points is bounded 
from above by
\be
   \label{eq:plane_upperbound}
   \N < \prod_{c=1}^{L-N+1}  2  \floor{\frac{N_c}{4}} + 1 ,
\ee
where $N_c$ is the number of nodes in the cycle $c$.
\end{thm}

\begin{proof}
In a planar network no two different fixed points can have the same winding vector $\vec \varpi$ (see lemma \ref{lem:wind-fp-corr-planar}) such that we can just count the different allowed winding vectors.
For each fundamental cycle $c \in B_C$ we have
\be
    - \floor{N_c/4} < 
     \varpi_c = \frac{1}{2 \pi} \sum_{e \in \, {\rm cycle} \, c} \Delta_e 
     < + \floor{N_c/4} 
\ee
because  $-\pi/2 < \Delta_e < +\pi/2$ in normal operation. Counting the number of different possible values of the winding numbers $\varpi_1,\ldots,\varpi_{L-N+1}$ respecting these upper and lower bounds yields the result.
\end{proof}

\subsubsection{Asymptotic behaviour}
We have shown in subsection \ref{subsec:complexnet} that it is not straightforward 
to obtain bounds for the number of fixed points $\N$ in complex networks, unlike simple 
cycles. However, in the limit of $N \gg 1$, we can nevertheless derive the 
scaling behaviour for $\N$.   

\paragraph{Two-cycle network}

\begin{figure}
\begin{center}
        \includegraphics[width=0.5\textwidth]{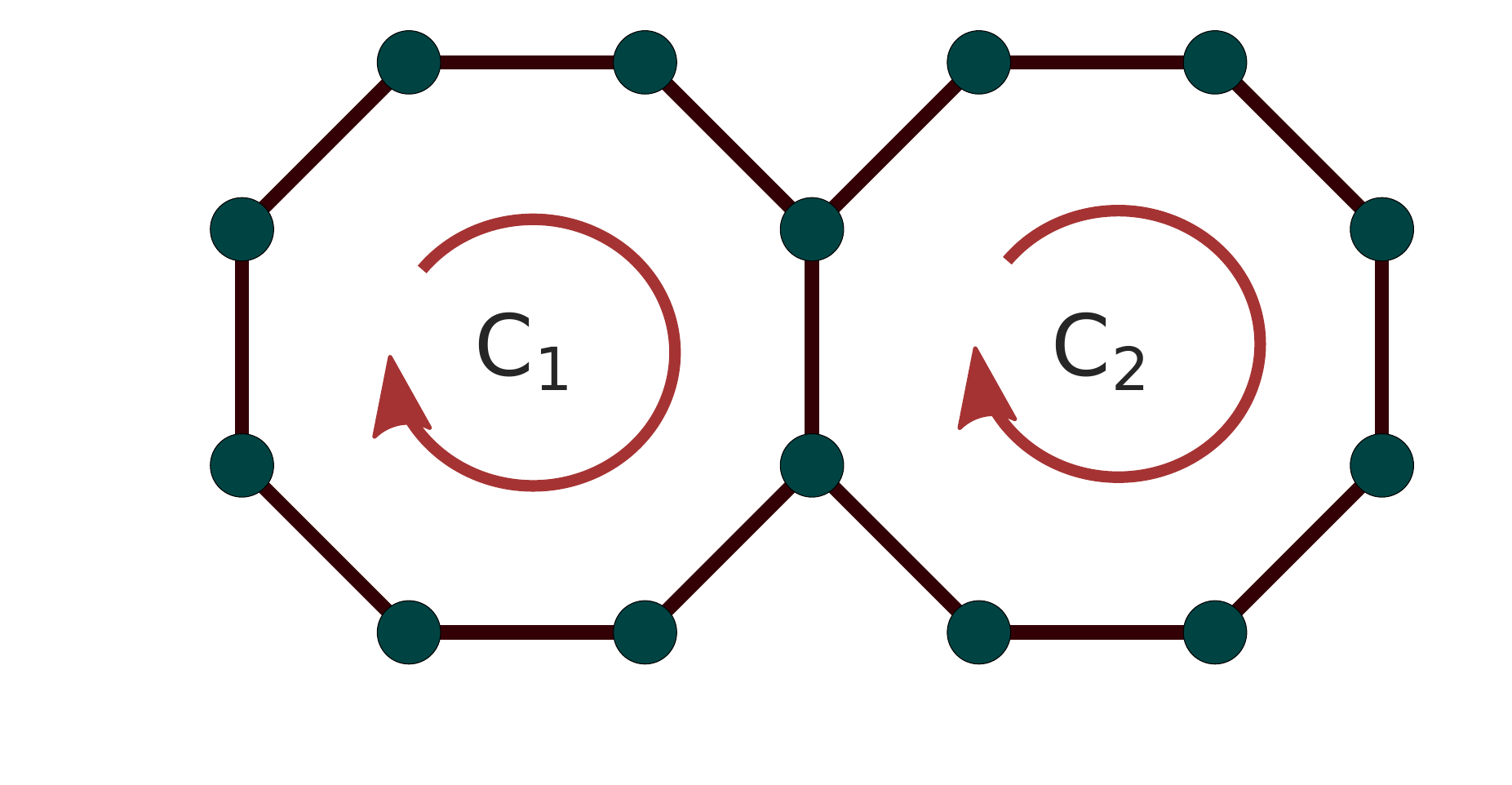}
        \caption{A 2-cycle network. We use the convention that cycles are 
        counter-clockwise.  Therefore we assign positive magnitude to 
        counter-clockwise cycle flows
        and negative magnitude to clockwise cycle flows.}
        \label{fig:2cycle}
\end{center}
\end{figure}

For simplicity, we first consider a network with homogeneous transmission 
capacities consisting of two cycles $C_1$ and 
$C_2$, as illustrated in Figure
\ref{fig:2cycle}. Suppose there are $n_1$ edges belonging only to cycle 1,
$n_2$ edges belonging only to cycle 2 and $n_{12}$ edges belonging to both. 
Let one fixed point be $\vec{\theta^*}$ with flows in each cycle and the intersection be bounded by
\begin{align}
    F_{1}^{\min} & \leq \min_{e\in C_1} F_{e}  \leq F_1^{\max}\\
    F_{2}^{\min} & \leq \min_{e\in C_2} F_{e}  \leq F_2^{\max}\\
    F_{12}^{\min} & \leq \min_{e\in C_2\cap C_1} F_{e}  \leq F_{12}^{\max}.
\end{align}

Then the possible cycle flows in each cycle are bounded inside a convex polygon $\cfdom$
described by 
\begin{align}
    \label{eq:diamond2d}
    -K + F_1^{\rm min} & \leq  f_1  \leq   K - F_1^{\rm max}\\
    -K + F_2^{\rm min} & \leq  f_1  \leq   K - F_2^{\rm max}\\
    -K + F_{12}^{\min} & \leq  f_1 - f_2  \leq K - F_{12}^{\max}.
\end{align}

Then, for $K \gg 1, n_1 \gg 1, n_2 \gg 1$, the number of fixed points 
converges to the area of the image set of $\cfdom$ under the mapping $\vec{\varpi}$.  
\begin{align}
    \label{}
    \N & \approx \int_{\rm \vec{\varpi}(\cfdom)} d\varpi_1 d\varpi_2 \\
& = \int_{\cfdom} df_1 df_2 \det J(\vec{\varpi}),
\end{align}
where the Jacobian $J(\vec{\varpi})$ for the change of variable can be computed from the 
expression for $\vec{\varpi}$ in \eqref{eq:def-winding-normal}, which yields
\begin{align}
    \det J(\vec \varpi)& = {\tiny \frac{1}{4K^2\pi^2}\det \begin{pmatrix}
    \displaystyle \sum_{e \in C_1}\frac{1}{\sqrt{1-\left(\frac{F_e+f_1}{K}\right)^2}} 
    + \displaystyle \sum_{e \in C_1\cap C_2}\frac{1}{\sqrt{1-\left(\frac{F_e+f_1-f_2}{K}\right)^2}} & 
    -\displaystyle \sum_{e \in C_1\cap C_2}\frac{1}{\sqrt{1-\left(\frac{F_e+f_1-f_2}{K}\right)^2}} \\
    -\displaystyle \sum_{e \in C_1\cap C_2}\frac{1}{\sqrt{1-\left(\frac{F_e+f_1-f_2}{K}\right)^2}} &
     \displaystyle \sum_{e \in C_2}\frac{1}{\sqrt{1-\left(\frac{F_e+f_2}{K}\right)^2}} 
    + \displaystyle \sum_{e \in C_1\cap C_2}\frac{1}{\sqrt{1-\left(\frac{F_e+f_1-f_2}{K}\right)^2}}.
    \end{pmatrix}
}
\end{align}

Taking the limits
\begin{align}
    \lim_{K\to\infty} \frac{F_e+f_1}{K} = \frac{f_1}{K}\nn \\
    \lim_{K\to\infty} \frac{F_e+f_2}{K} = \frac{f_2}{K},
\end{align}
leads to
\begin{align}
\label{eq:N_afterlimit}
\N & \approx \frac{1}{4K^2\pi^2}\int_{\cfdomlim} df_1 df_2 \det \begin{pmatrix}
    \frac{n_1}{\sqrt{1-\left(\frac{f_1}{K}\right)^2}} 
    + \frac{n_{12}}{\sqrt{1-\left(\frac{f_1-f_2}{K}\right)^2}} & 
    - \frac{n_{12}}{\sqrt{1-\left(\frac{f_1-f_2}{K}\right)^2}} \\
    - \frac{n_{12}}{\sqrt{1-\left(\frac{f_1-f_2}{K}\right)^2}} &
    \frac{n_2}{\sqrt{1-\left(\frac{f_2}{K}\right)^2}} 
    + \frac{n_{12}}{\sqrt{1-\left(\frac{f_1-f_2}{K}\right)^2}} 
    \end{pmatrix}\\
    \cfdomlim & :=\{(f_1,f_2): (f_1, f_2)\in\mathbb{R}^2, |f_1|\le K, |f_2|\le K, |f_1-f_2|\le K\}.\nn
\end{align}
Redefining $f_1 \to f_1/K, f_2 \to f_2/K$, we obtain
\begin{align}
\N & \approx \frac{1}{4\pi^2}\int_{\cfdomlim} df_1 df_2 \det \begin{pmatrix}
    \frac{n_1}{\sqrt{1-f_1^2}} 
    + \frac{n_{12}}{\sqrt{1-\left(f_1-f_2\right)^2}} & 
    - \frac{n_{12}}{\sqrt{1-\left(f_1-f_2\right)^2}} \\
    - \frac{n_{12}}{\sqrt{1-\left(f_1-f_2\right)^2}} &
    \frac{n_2}{\sqrt{1-f_2^2}} 
    + \frac{n_{12}}{\sqrt{1-\left(f_1-f_2\right)^2}} 
    \end{pmatrix} \\
& = \frac{1}{4\pi^2}\left( n_1n_2\int_{\cfdomlim} \frac{1}{\sqrt{1-f_1^2}}\frac{1}{\sqrt{1-f_2^2}} df_1 df_2  
    + n_1n_{12}\int_{\cfdomlim} \frac{1}{\sqrt{1-f_1^2}}\frac{1}{\sqrt{1-(f_1-f_2)^2}}df_1 df_2 \right.\\ \nn
    & \left.\qquad \qquad + n_2n_{12}\int_{\cfdomlim} \frac{1}{\sqrt{1-f_2^2}}\frac{1}{\sqrt{1-(f_1-f_2)^2}}df_1 df_2  \right)\\
& = \frac{1}{4\pi^2}\left( n_1n_2\int_{\cfdomlim} \frac{1}{\sqrt{1-f_1^2}}\frac{1}{\sqrt{1-f_2^2}} df_1 df_2  
    + (n_1+n_2)n_{12}\int_{\cfdomlim} \frac{1}{\sqrt{1-f_1^2}}\frac{1}{\sqrt{1-(f_1-f_2)^2}}df_1 df_2 \right). 
\end{align}

In the last line we use the symmetry in $f_1, f_2$, both in the integrand, and the domain
of integration. We can simplify even further, by using the following 
change of variables in the second integral:
\begin{align*}
(f_1, f_2) \mapsto (f_2 - f_1, f_2).
\end{align*}
We note that the domain remains the same  after this change of variable, and 
the determinant of the Jacobian $\det(J) = -1$. 
This allows the simplification
\begin{align*}
\N & \approx \left( n_1n_2 + (n_1+n_2)n_{12}\right)\underbrace{\frac{1}{4\pi^2}\int_{\cfdomlim}
\frac{1}{\sqrt{1-f_1^2}}\frac{1}{\sqrt{1-f_2^2}} df_1 df_2}_{\gamma},
\end{align*}
with
\begin{align*}
\gamma & = \frac{1}{4\pi^2}\left\{\int_{-1}^0\frac{df_1}{\sqrt{1-f_1^2}}\int_{-1}^{f_1+1}\frac{df_2}{\sqrt{1-f_2^2}}
+ \int_{0}^1\frac{df_1}{\sqrt{1-f_1^2}}\int_{f_1-1}^{1}\frac{df_2}{\sqrt{1-f_2^2}}\right\}\\
& = \frac{1}{4\pi^2}\left\{\int_{-1}^0\frac{\arcsin(f_1+1)+\frac{\pi}{2}}{\sqrt{1-f_1^2}}df_1
+ \int_0^1\frac{\frac{\pi}{2}-\arcsin(f_1-1)}{\sqrt{1-f_1^2}}df_1 \right\}\\
& = \frac{1}{4\pi^2}\left\{\frac{\pi^2}{4} +\int_{-1}^0\frac{\arcsin(f_1+1)}{\sqrt{1-f_1^2}}df_1
+ \frac{\pi^2}{4} +\int_{0}^1\frac{\arcsin(f_1-1)}{\sqrt{1-f_1^2}}df_1\right\} \\
& = \frac{1}{8} + \frac{1}{2\pi^2} \int_{-1}^0\frac{\arcsin(f_1+1)}{\sqrt{1-f_1^2}}df_1\\
& \approx 0.1576,
\end{align*}
to finally yield this scaling result
\begin{align}
    \label{eq:scaling_2cycle}
    \lim_{n_1,n_2\to\infty}\N =& \gamma \left( n_1n_2 + (n_1+n_2)n_{12}\right)\\
    \gamma &\approx 0.1576.\nn
\end{align}

\begin{figure}[!htp]
\begin{center}
    \includegraphics[height=0.27\textheight]{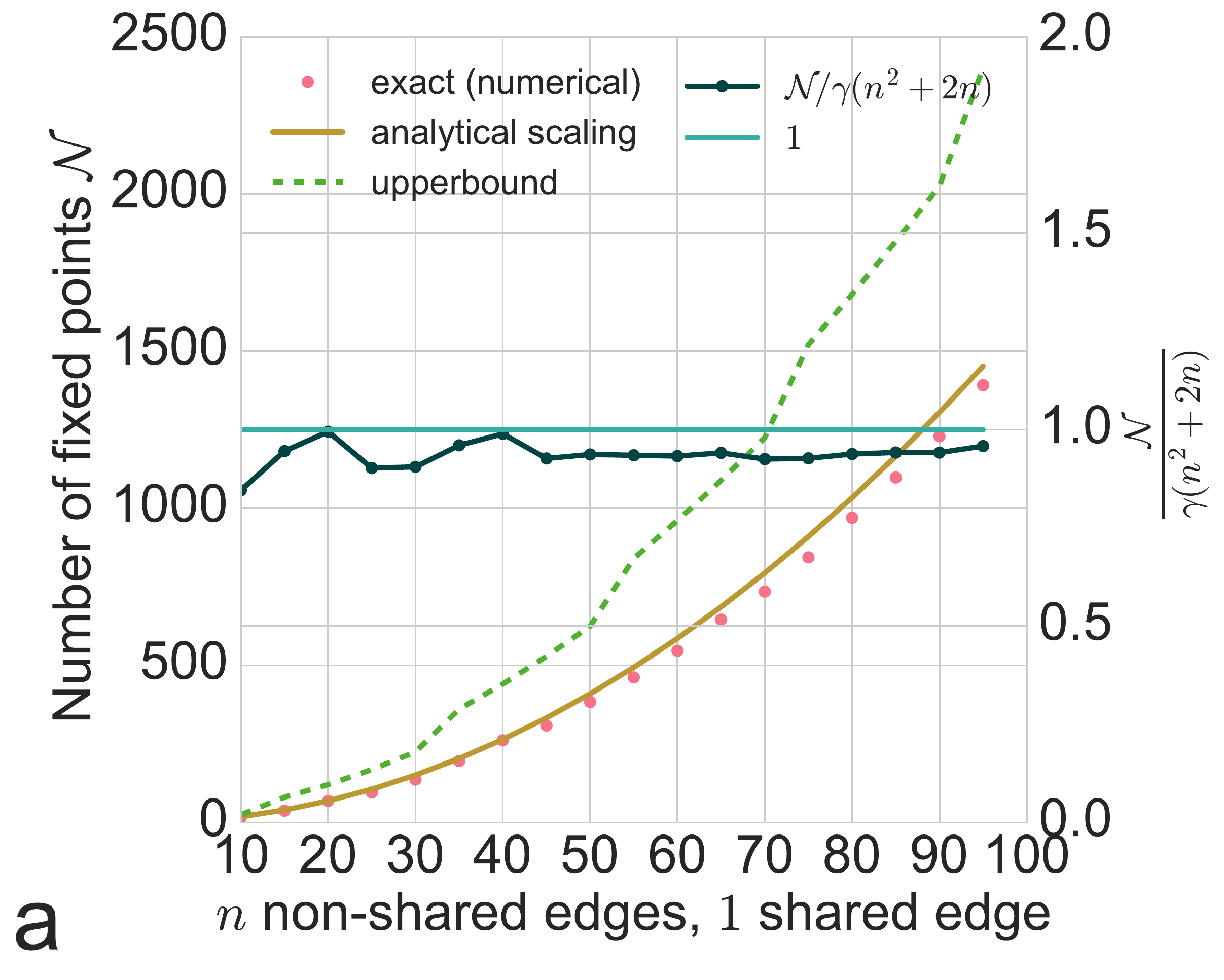}
    \includegraphics[height=0.27\textheight]{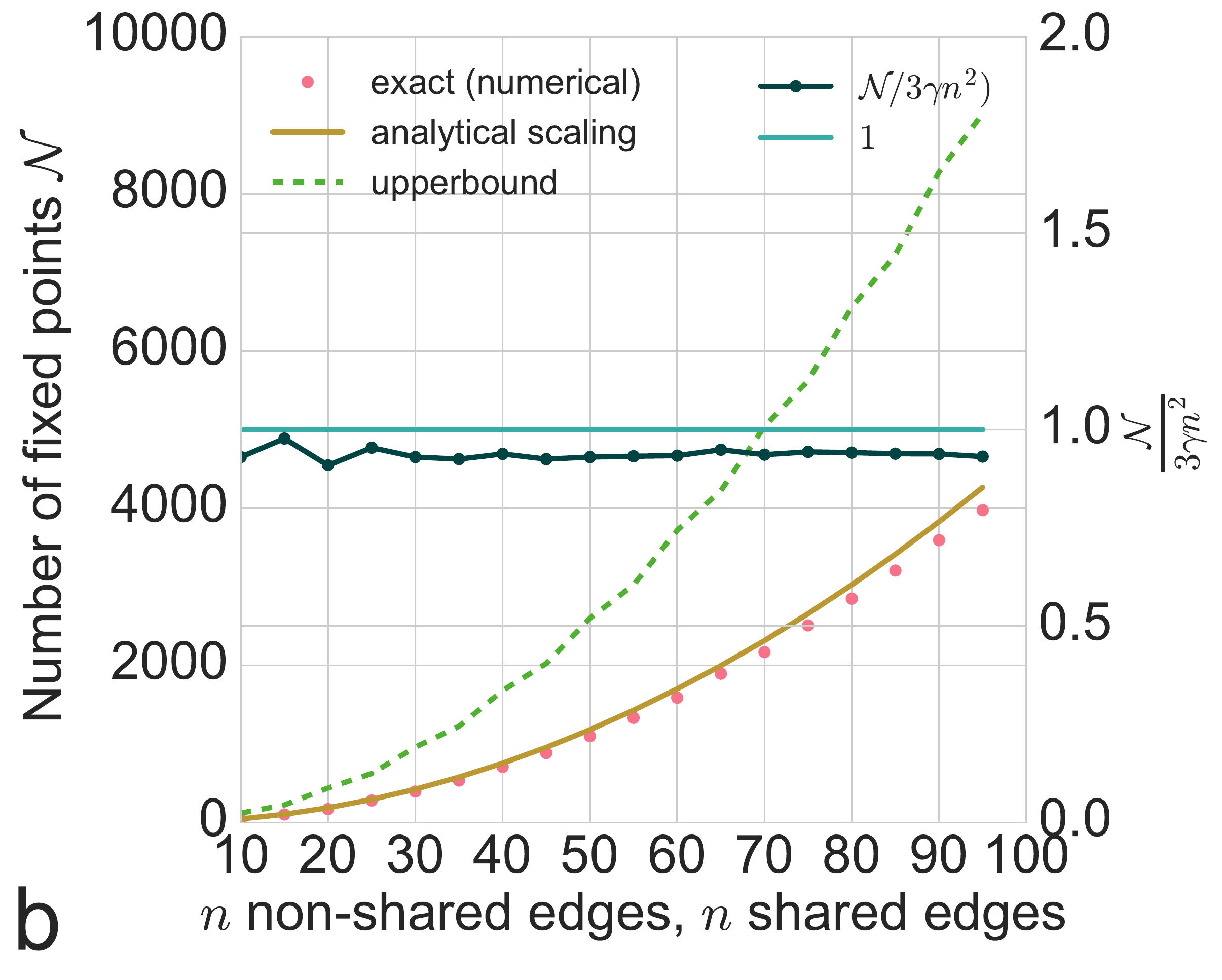}
    \caption{(a) Scaling of number of fixed points for two-cycle networks. Each cycle
        has $n+1$ edges and they share one edge between them.   (Left 
y-axis)   The 
dots show the exact number of fixed points computed numerically.  The solid 
line shows the predicted number as per scaling relation \eqref{eq:scaling_2cycle}. The dashed 
line shows the upperbound \eqref{eq:plane_upperbound}. (Right y-axis) The dot-dashed line shows 
the number of fixed points divided by $n^2+2n$ converging to a constant value, 
that is close to the analytically predicted value $\gamma = 0.1576$, as per 
equation \eqref{eq:scaling_2cycle}. (b) Same as (a), but for networks where each cycle
        has $2n$ edges and they share $n$ edges between them. }
\label{fig:scaling_2cycle}
\end{center}
\end{figure}

To evaluate the accuracy of \eqref{eq:scaling_2cycle}, we apply it 
to two special cases.  First, we consider networks with $n = n_1 = n_2, n_{12} = 1$, i.e. two identical cycles 
sharing only one single edge. In this case \eqref{eq:scaling_2cycle} becomes
\begin{align}
    \label{eq:scaling_nn1}
    \N(n, n, 1) & \approx (n^2 + 2n)\gamma.  
\end{align}
Second, we consider networks with $n = n_1 = n_2, n_{12} = n$, i.e. two identical cycles 
sharing half of their edges. In this case \eqref{eq:scaling_2cycle} becomes
\begin{align}
    \label{eq:scaling_nn1}
    \N(n, n, n) & \approx 3 \gamma n^2.  
\end{align}

We see in Figure 
\ref{fig:scaling_2cycle} that in both these two cases, the scaling relations are quite accurate even for not very large 
network sizes, such as $n=50$.

\subsubsection{General planar graphs}
The scaling results for two-cycle networks can be extended to general planar graphs, to this end we define a 
few quantities. 

\begin{figure}
\begin{center}
    \includegraphics[width=0.4\textwidth]{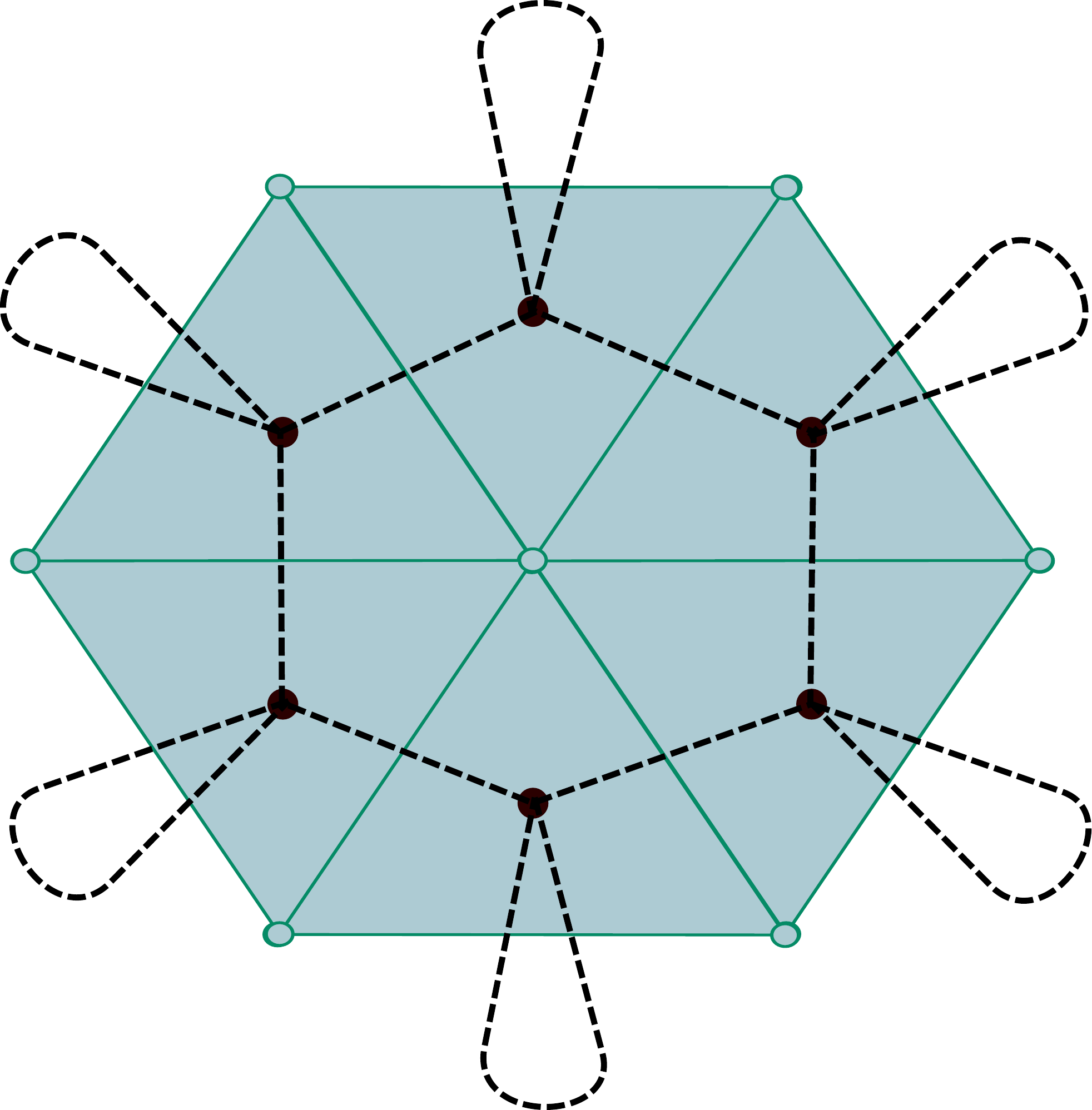}
    \caption{A planar graph (solid edges, unfilled circular nodes) and its loopy 
    dual (dashed edges, filled circular nodes) corresponding to this specific 
    embedding.  }
    \label{fig:dualgraph}
\end{center}
\end{figure}

\begin{defn}[Loopy dual graph]
    Given a planar graph $G(V, E)$ and an embedding, we choose a cycle basis
$B_C$ consisting of the faces of the embedding. The loopy dual graph 
$\DualGraph{G}$ is a undirected multigraph whose vertex set is equal to $B_C$. 
Its edge set $E'$ is as follows.  For each edge $e\in E$, if it is 
shared between two cycles $c_1$ and $c_2$, then an edge between $c$ and $c'$ 
is added to $E'$. If $e$ is at the boundary and belongs to only one 
cycle $c$, then a self loop is added at node $c$.  We illustrate this definition in Figure 
\ref{fig:dualgraph}. 
\end{defn}

Now, consider a planar graph and an arbitrarily chosen fixed point with
flows $F_e$.  Let's denote by $\LoopyLapl$ the loopy laplacian of its meta 
graph, as defined in Definition \ref{def:mata-graph}.

Then equation \eqref{eq:N_afterlimit} generalizes to 
\begin{align}
    \N & \approx \frac{1}{(2K\pi)^{L-N+1}}\int_{\cfdomlim} df_1 df_2\cdots df_{L-N+1} \det{\LoopyLapl}\\
    \cfdomlim & :=\{(f_1,f_2,\cdots f_{L-N+1}): |f_i|\le K, |f_i-f_j|\le K 
\text{ if cycles  }i,j  \text{ share an edge}\}.\nn
\end{align}

\section{Unstable fixed points}
\label{sec:unstable}

In principle, we can generalize the cycle flow approach to find fixed points which do not satify the normal operation condition, too. These fixed points are typically linearly unstable (cf.~the discussion in \cite{epjst14}). 
However, most of the results on the number of fixed points cannot be generalized to this case. As an instructive example, consider again a homgeneous ring as in section \ref{eqn:sec-cycle}. We label the nodes as $1,2,\ldots,N$ along the cycle and assume that $N$ is an integer multiple of $4$. All nodes have a vanishing power injection $P_j \equiv 0$ and all links have an equal strength $K$ as before. Then it is easy to see that 
\be
   \vec \theta^* = (0,\delta,\pi,\pi+\delta,2\pi, 2\pi+\delta,3\pi,\ldots)^T
\ee
is a fixed point of the dynamics for each value of $\delta \in [0,\pi)$. This class of fixed points represents a pure cycle flow,
\begin{align}
   F_{j,j+1} &= K \sin(\theta_{j+1} - \theta_j) = K \sin(\delta) 
\end{align}
for all edges $(j,j+1)$. The winding number is $\varpi = N/4$ independent of the value of $\delta$ and the edges belong alternately to $E_+$ and $E_-$:
\begin{align}
   E_+ &= \big\{ (1,2); (3,4); (5,6); \ldots \big\}  \nn \\
   E_- &= \big\{ (2,3); (4,5); (6,7); \ldots \big\} .
\end{align}

This simple example shows that two main assumptions made for the normal operation fixed points 
(where $E_- = \emptyset$) do not longer hold:  First, the set of fixed points is no longer discrete. 
Instead we find a continuum of solutions parametrized by the real number $\delta$.
Second, different fixed points yield the same winding number. Thus we cannot obtain the 
number of fixed points by counting winding numbers in general.

\section{Calculating fixed points}

The cycle flow approach yields a convenient method to calculate multiple fixed points for oscillator networks. Generally, it is hard to make sure that a numerical algorithm yields all solutions for a nonlinear algebraic equation. However, we have shown that the winding numbers are unique at least for normal operation fixed points in planar networks. Thus we can scan the allowed values of the winding numbers and try to find a corresponding solution. This can be done by starting from an arbitrary solution of the dynamical condition and adding cycle flows until we obtain the desired winding numbers.

\begin{figure}[!tb]
\centering
\includegraphics[width=0.8\columnwidth]{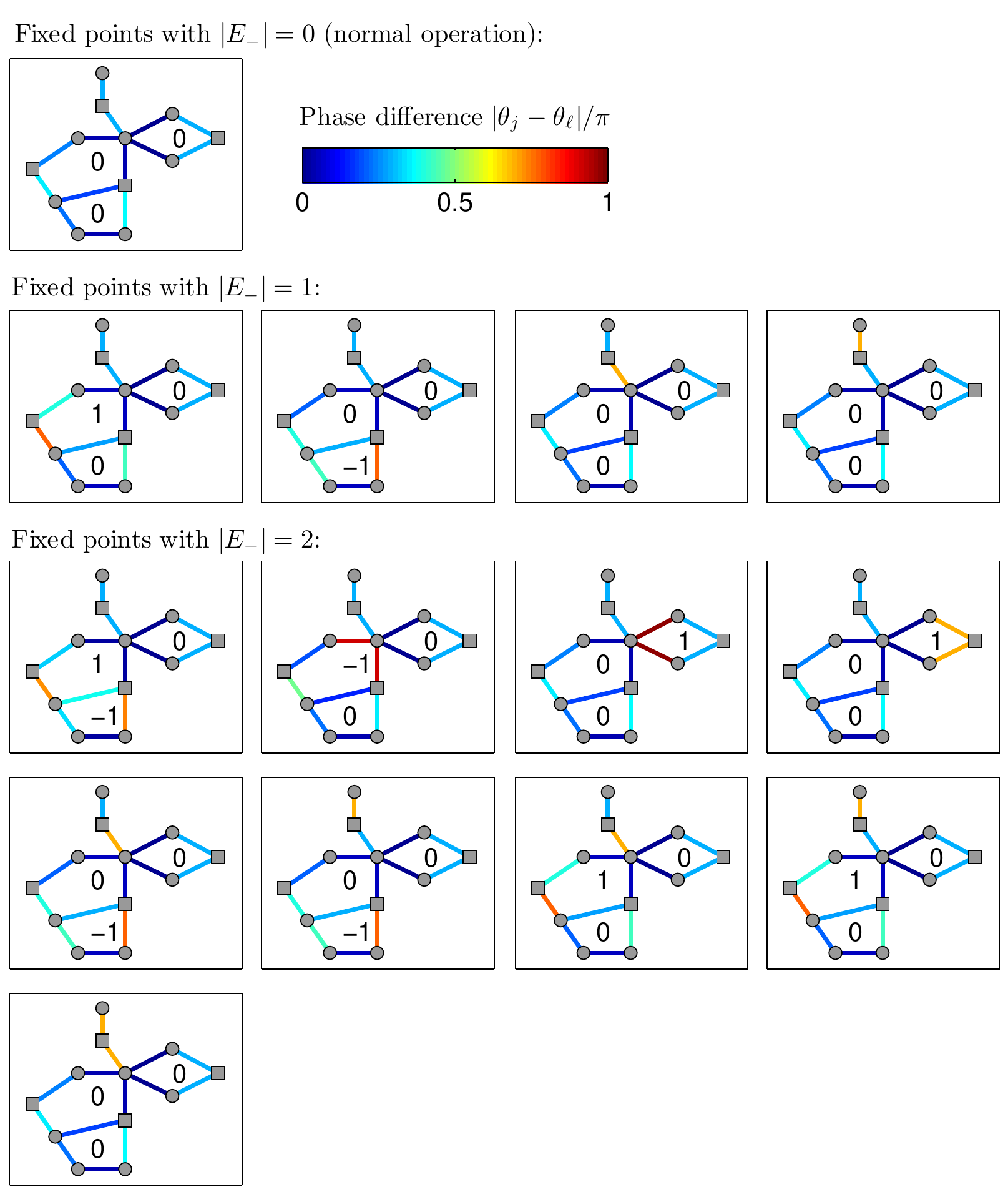}
\caption{\label{fig:moritznet}
All fixed points with $|E_-|\le 2$ in a network with three cycles calculated using the algorithm described in the main text. The winding number of each cycle is displayed.
Squares represent generators with $P = +2P_0$, circles consumers with $P=-P_0$. All links have a coupling strength of $K = 24/19 \times P_0$.
}
\end{figure}

In particular, we can calculate all fixed points in normal operation for a planar network using 
the following algorithm:
\begin{enumerate}
\item Find a solution $F^{(0)}$ of the dyanmic condition.
\item Fix a plane embedding and a cycle basis.
\item Vary the number $z_c$ in the interval $[-\frac{N_c}{4}, \frac{N_c}{4}]$, for all 
cycles $c=1,\ldots,L-N+1$. 
\item Try to solve the set of equation 
\be
    \varpi_c(\vec f) = z_c  \qquad \text{ for all } \; c = 1,\ldots,L-N+1 
\ee    
where the winding numbers are given by equation (\ref{eq:def-winding}).
\end{enumerate}

Dropping the assumption of a normal operation, we loose the the guarantee of 
uniqueness as discussed in section \ref{sec:unstable}. Nevertheless the method can be 
readily adapted to find \emph{most} of the unstable fixed points, at least if 
the number $|E_{-}|$ is small. This can 
be very useful, as a systematic calculation of such fixed points is 
generally not straightforward. The results can be applied, among other 
things, to assess the global stability of a stable fixed point by analyzing the 
stability boundary \cite{menck2013basin,chiang1987foundations} or the stability in the presence of stochastic 
fluctuations \cite{moritz16escape}.
In particular, we must add another step to the algorithm to loop over all possible sets $E_-$:
\begin{enumerate}[label=4.\alph*.]
\item Vary $k=0,\ldots,L$. Then sample all $k$-tuples from the edge set $E$ to define the set $E_-$.
\item Vary the number $z_c$ in the interval $[-\frac{N_c}{4}, \frac{N_c}{4}]$, for all 
cycles $c=1,\ldots,L-N+1$. 
\end{enumerate}

The output of this algorithm is shown in Figure \ref{fig:moritznet} for a small test network 
and $|E_-| \le 2$. For this small network we have only $L-N+1=3$ fundamental 
cycles of which one is decoupled. Hence we can graphically check that we 
hvae obtained \emph{all} fixed points.

\section{Discussion}

Oscillator networks are ubiquitous in nature and technology.
A lot of research in statistical physics starting from Kuramoto's seminal work \cite{Kura75} 
has been devoted to the onset of partial synchronization in large networks. However, in
some applications global synchronization is required. In particular in electrical power grids
all generators have to run with exactly the same frequency and have to be strictly 
phase-locked to enable stable power flows to the customers. A desynchronization
generally has catastrophic consequences. An example is provided by the European
power blackout in November 2006. Following a shutdown of one transmission line
and unsuccessful attempts to restore stable operation, the European grid fragmented
in three mutually asynchronous clusters \cite{UCTE07}. In the end more than 10 million customers 
were cut from the power supply.

In this article we have analyzed the existence of stable fixed points in finite oscillator 
networks. The main methodological advancement is to split the calculation into two parts:
First, we calculate the flows which satisfy the continuity equation at all nodes. Then we
single out the the specific solution which leads to consistent phases of the oscillators. 
We thus move the focus of the calculation from the nodes (phases) to the edges
(flows) and cycles. An immediate consequence is that several fixed points can 
co-exist which differ by cycle flows. Thus oscillator networks are in general multistable. 

For networks containing a single cycle we have obtained upper and lower 
bounds for the number of fixed points in terms of three structural quantities:
the maximal partial net power $\bar P_{\rm max}$, which measures the homogeneity 
of the power injections or natural frequencies, respectively, and the maximum and 
minimum edge strength along the cycle. We find that generally the number of stable 
fixed point is particularly large if (a) the cycle is long, (b) the edge strength are 
large and (c) the power sources are distributed homogeneously. However, the example
discussed in section \ref{sec:braess}
 shows that extreme care has to be taken for special network topologies.
Increasing the strength of the wrong edge can also decrease the number of
fixed points.   
Finding bounds for the number of stable fixed points in general network topologies
is much more involved. Results have been obtained for planar networks, but the bounds
are much weaker as for networks with single cycles.
Interestingly,  both  tree networks and fully connected networks have at most one stable
fixed points. However, networks with intermediate sparsity, which is most 
realistic for electrical power grids, may exhibit multistability.  

Several aspects of multistability have been previously discussed in the literature. 
Multistability in isolated rings was discussed in \cite{Ocha09}. The limits 
(\ref{eq:Nring1}) were derived and the basins of attraction of the different fixed points
was studied numerically.
The case of a densely connected graph was analyzed by Taylor in \cite{Tayl12}.
He was able to show that there is at most one stable fixed point if the node
degree is at least $0.9395 \times (N-1)$.  
Mehta et al. investigate multistability in complex networks numerically using a 
similar approach as the present paper \cite{Meht14}. They argue that the number of fixed
points scales with the number of cycles as each cycle can accommodate cycle flows.
While this is valid for many graphs, there are counterexamples (Figure 
\ref{fig:lowerbound-nosol}). Delabays et al \cite{delabays15multistab} have recently reported their 
treatment of multistability using cycle flows. They have extended the upper bounds for fixed points in single 
rings by \cite{Ocha09} to include also those stable fixed points with phase differences
along edges $>\pi/2$.  They have also derived  upper bounds \cite{delabays16} for number of 
fixed points for planar graphs in case of uniform power injections at all 
nodes.  Xi et al. \cite{xi2016nonlinstab} have numerically shown that spatical heterogeneity
of power injections $P_j$ reduces the number of fixed points, which fits with our analytical
result in Corollary \ref{corr:hom-ring}.   Intriguingly, they have also found that
in heterogeneous ring topologies, nonlinear stability of fixed points decrease 
with the ring size $N$. 

In this work, we have obtained a \emph{lower bound} for the number of 
fixed 
points, and thereby provided a \emph{sufficient condition} for existence of 
multistability.  Furthermore, we have shown that the 
length of the cycles $N_c$ and the homogeneity $\bar P_{\rm max}$ are equally 
important for multistability and thereby arrived at tighter bounds for the number 
of fixed points than Ochab and Gora \cite{Ocha09} and Delabays et al 
\cite{delabays15multistab}. Morever, we have derived scaling laws at the limit of
infinite transmission strengths that are much tighter than the upper bound 
results previously reported.  We have shown the derived scaling behaviour to 
match numerically computed exact results for moderately sized networks.  

Interestingly, Our results show that a previous highly recognized result presented by 
Jadbabaie et al in \cite{Jadb04} is incorrect. The authors claim that
for any network of Kuramoto oscillators with different natural frequencies, 
there exists a $K_{u}$ such that for $K>K_u$ there is only one stable fixed point.
This claim is disproven by the examples presented in section (\ref{eqn:sec-cycle}) 
as well as by the rigorous results on the existence of multiple fixed points in 
corollary \ref{corr:hom-ring}.
The error in the proof of \cite{Jadb04} is rather technical. The authors
define a function $\vec L$ such that the defining equation of a fixed point 
(\ref{eqn:def-steady-state}) can be rewritten in the form
\be
\vec{\theta}^*  = \vec{L}(\vec{\theta}^*) .
    \label{eq:banach}
\ee
Jadbabaie et al. then claim that $\vec{L}$ is a contraction on the subset 
of $\vec{\theta}$ such that $|\theta_i - \theta_j|<\pi/2$ for all edges $(i,j)$, which we called
normal operation. Banach's contraction theorem then yields that the 
algebraic equation (\ref{eq:banach}) has a unique fixed point. The problem 
is that the range of $\vec L(\theta)$ is generally \emph{not} a subset of the subspace
of normal operation, even if the domain is. After applying $\vec L$, some phase differences can get 
out of the interval $[-\pi/2, \pi/2]$. Thus Banach's contraction theorem cannot be applied, which
spoils the proof.

\section{Conclusion}
In summary, taking cycle flows as a basis of flow patterns we analyzed existence
and stability of phase locked states in networks of Kuramoto oscillators and second
order phase oscillators modeling the phase dynamics of electric power grids.
We demonstrated that such systems exhibit multistability. Intriguingly, multistability
prevails even under conditions where unique stable operating points were believed to
exist in both a power engineering text book and a major complex network reference on
Kuramoto oscillators \cite{Jadb04, Mach08}. For classes of network topologies, we have established
necessary and sufficient conditions for multistability, derived lower and upper bounds
for the number of fixed points. We explained why generalizing those bounds for arbitrary
topologies is hard. Nevertheless, we have derived asymptotic scaling laws at large loop limit that has 
been found to match closely numerically obtained exact results.  

\section{Acknowledgements}
We gratefully acknowledge support from the Federal Ministry of Education and 
Research (BMBF grant no.~03SF0472A-E), the Helmholtz 
Association (via the joint initiative `Energy System 2050 -- A Contribution
of the Research Field Energy' and the grant no.~VH-NG-1025 to D.W.) and the
Max Planck Society to M.T. The works of D.M. is supported by the IMPRS
Physics of Biological and Complex Systems, G\"ottingen.


\bibliography{frustration}

\end{document}